\documentclass[12pt,draftclsnofoot,onecolumn]{IEEEtran}

\usepackage{cite} 
\usepackage{xspace}
\usepackage{filecontents}

\usepackage{float}
\usepackage{pgf, tikz}

\usepackage{graphicx}

\usepackage[cmex10]{amsmath}

\usepackage{bm}
\usepackage{soul,color} 
\usepackage{amssymb} 

\usepackage{array}

\usepackage{subfigure}
\usepackage{caption2} 

\usepackage{amsfonts}
\usepackage{amsthm}   

\usepackage{enumerate} 

\usepackage{extarrows}

\usepackage{tkz-fct}

\newtheorem{lemma}{Lemma}
\newtheorem{proposition}{Proposition}

\newtheorem{definition}{Definition}
\usetikzlibrary{patterns,snakes}

\usepackage{enumerate}

\usepackage{extarrows}

\usepackage{multicol} 

\newcommand{\myexpect}[1]{\mathbb{E}\left( #1 \right)}

\newcommand{\myprobability}[1]{\mathrm{Pr}\left( #1 \right)}

\newcommand{\sigmaone}{\sigma^2_{\mathrm{cov}}}
\newcommand{\sigmatwo}{\sigma^2_{\mathrm{rec}}}

\newcommand{\sigmaonesqrt}{\sigma_{\mathrm{cov}}}
\newcommand{\sigmatwosqrt}{\sigma_{\mathrm{rec}}}

\newcommand{\myomegagen}{\Omega_{\mathrm{gen}}}
\newcommand{\myomegagennew}{\breve{\Omega}_{\mathrm{gen}}}

\newcommand{\mypower}{\mathcal{P}}
\newcommand{\myP}{\mathcal{P}}
\newcommand{\Ei}{\mathrm{Ei}}

\newcommand{\snrcov}{\mathsf{SNR}_{\mathrm{coh}}}
\newcommand{\snrrec}{\mathsf{SNR}_{\mathrm{non-coh}}}
\newcommand{\snr}{\mathsf{SNR}}

\begin{document}
\title{A Novel Receiver Design with Joint Coherent and Non-Coherent Processing}
\author{Wanchun Liu, Xiangyun Zhou, Salman Durrani and Petar Popovski   \thanks{\setlength{\baselineskip}{13pt} \noindent Wanchun Liu, Xiangyun Zhou and Salman Durrani are with Research School of Engineering, The Australian National University, Canberra, ACT 2601, Australia (emails: \{wanchun.liu, xiangyun.zhou, salman.durrani\}@anu.edu.au). 
Petar Popovski is with the Department of Electronic Systems, Aalborg University, Aalborg 9220, Denmark (email: petarp@es.aau.dk). }}
\maketitle

\begin{abstract}
In this paper, we propose a novel splitting receiver,
which involves joint processing of coherently and non-coherently received signals.
Using a passive RF power splitter, the received signal at each receiver antenna is split into two streams which are then processed by a conventional coherent detection (CD) circuit and a power-detection (PD) circuit, respectively.
The streams of the signals from all the receiver antennas are then jointly used for information detection.
We show that the splitting receiver creates a three-dimensional received signal space, due to the joint coherent and non-coherent processing.
We analyze the achievable rate of a splitting receiver, which shows that the splitting receiver provides a rate gain of $3/2$ compared to either the conventional (CD-based) coherent receiver or the PD-based non-coherent receiver in the high SNR regime.
We also analyze the symbol error rate (SER) for practical modulation schemes, which shows that the splitting receiver achieves asymptotic SER reduction by a factor of at least $\sqrt{M}-1$ for $M$-QAM compared to either the conventional (CD-based) coherent receiver or the PD-based non-coherent receiver.
\end{abstract}

{\hspace{6.8cm} \textbf{\small Index Terms}}\par
{\small Receiver~design; energy/power-detection based communications; joint coherent and non-coherent processing; 
	\\ \hspace*{0.4cm}RF signal splitting.}
\newpage
\section{Introduction}
\subsection{Background and Motivation}
Wireless communications has witnessed several major theoretical advancements in the last few decades, which have been quickly incorporated into communication standards, e.g., multiple input multiple output (MIMO) systems and orthogonal frequency division multiple access (OFDMA)~\cite{BOOKTse}.
The baseband receiver design, underlying these technologies, has been the receiver based on a coherent detection (CD), i.e., coherent receiver for short, which has been adopted exclusively in nearly all the popular wireless communication standards.
The coherent receiver design has {\color{black}stayed} virtually unchanged throughout the evolution of wireless communication systems to date.

The recent trend towards a large number of antennas at the transmitter or receiver (such as in massive MIMO and millimeter wave systems~\cite{MagazineRobert,Jeffrey}) provides incentive to rethink the coherent receiver design for future communication systems.
This is because provision of accurate channel state information (CSI) becomes challenging in scenarios with a massive number of antennas.
In this regard, \emph{power/intensity-detection} (PD) based \emph{non-coherent} receivers have been proposed~\cite{Goldsmith16} and~\cite{Petar16}.
Although, in general, non-coherent receivers suffer from a performance loss compared to coherent receivers,
their low cost and low power consumption makes them attractive for systems with a large number of antennas.
It is proved in~\cite{Goldsmith16} that PD-based non-coherent modulation in a massive single input multiple output (SIMO) system can achieve a scaling law which is the same as the coherent modulation scheme with an increasing number of antennas.
It is also shown in \cite{Petar16} that the performance of the PD-based non-coherent modulation asymptotically
approaches the performance of coherent detection in high SNR regimes.
These studies show that PD-based non-coherent receivers can offer a performance comparable to coherent receivers in future wireless systems.

Another motivation for considering PD-based non-coherent receivers comes from the recent interest in simultaneous wireless information and power transfer (SWIPT)~\cite{Bi15,Huang15,Krikidis_survey,XiaoLu}.
In such systems, a user is assumed to be equipped with an energy receiver which is based on non-coherent RF (radio frequency)-to-DC (direct current) conversion, and a conventional coherent (information) receiver.
In SWIPT systems, the user employs the energy receiver and the coherent receiver \emph{separately}:
(i) In time-switching, the user switches between the two receivers, depending on whether it is in RF energy harvesting (EH) mode or information detection mode, or
(ii) in power-splitting, the user splits the received signal into two streams, and then sends one stream to the energy receiver and the other to the coherent receiver.

Although the PD-based non-coherent receiver has received more attention in wireless communication systems recently, it has played an important role in optical communications for a long time.
For the low-cost wireless infrared communication system~\cite{InfraredMag}, \emph{intensity modulation} (IM) is the most commonly adopted modulation scheme.
For IM-based wireless infrared communication, information is carried by the instantaneous power of the carrier, and the receiver uses a photodetector to produce a current proportional to the received instantaneous power directly, i.e., PD-based non-coherent demodulation.
The wireless infrared communication channel is usually referred to as the intensity channel or the \emph{non-coherent additive white Gaussian noise (AWGN) channel}, echoing the coherent AWGN channel.
The modulation schemes and the capacity of the non-coherent AWGN channel were studied in~\cite{OpticModu} and~\cite{OpticalTIT}, respectively.

\subsection{Novel Contributions}
Motivated by the recent interest in the PD-based non-coherent receiver,
we consider a basic point-to-point communication system and revisit the design of the communication receiver.
Rather than focusing on an improved design for either coherent receiver or non-coherent receiver alone,
we consider a receiver with joint coherent and PD-based non-coherent processing.
To the best of our knowledge, this is an open problem in the literature and it is not immediately clear whether joint processing will be better than either coherent or PD-based non-coherent processing alone.
In this work, we show that it can in fact significantly improve the achievable rate and also reduce the symbol error rate (SER).

The main contributions of the paper are summarized as follows:
\begin{enumerate} [1)]
\item We propose a novel information receiver architecture for a $K$-antenna receiver called \emph{splitting receiver}.
The received signal at each antenna is split into two streams by a passive power splitter with a certain \emph{splitting ratio}. 
One stream is processed by a conventional (coherent) CD circuit, and the other is processed by a (non-coherent) PD circuit, and then the $2 K$ streams of processed signal are jointly used for information detection. 
	
\item As a variant of the splitting receiver, we also propose a simplified receiver where no power splitters are required and a fixed number of antennas are connected to CD circuits and the remaining antennas are connected to PD circuits. 
Analytically, the simplified receiver can be treated as a special case of a splitting receiver, where the splitting ratio at each antenna can only take $1$ or $0$.

\item We show that the splitting receiver (and also the simplified receiver), increases the dimension of the received signal space, since the noise adds linearly to the signal in the coherent receiver part and the noise adds to the squared amplitude of the signal in the PD-based non-coherent receiver part.
This results in improved communication performance.

\item From an information-theoretic perspective, we model the channel introduced by the splitting receiver as a splitting channel.
Assuming a Gaussian input to the splitting channel, in the high signal-to-noise-ratio (SNR) regime, 
we show analytically that: 
(i) The asymptotic maximum mutual information of the splitting channel is $3/2$ times that of either 
the coherent AWGN channel or the non-coherent AWGN channel, under the same average received signal power constraint.
(ii) For a splitting receiver with a single receiver antenna, the asymptotic optimal power splitting ratio is $1/3$. 
(iii) For the simplified receiver with a large number of receiver antennas, connecting half the antennas to the CD circuits and the other half to the PD circuits is the optimal strategy. 

\item For transmissions based on practical modulations, we analyze the symbol decision region and the symbol error rate (SER) at the splitting receiver. 
Considering high SNR regime, we derive the SER expression for a general modulation scheme.
The analytical results show that, compared with the conventional coherent receiver, the splitting receiver achieves asymptotic SER reduction by a factor of $M-1$ for $M$-PAM (pulse amplitude modulation) and $\sqrt{M} -1$ for $M$-QAM (quadrature amplitude modulation).

\end{enumerate}

\subsection{Paper Organization and Notation}
This paper is organized as follows. 
Section II presents the system model, the proposed receiver architectures and the splitting channel. 
Section III analyzes the mutual information of the splitting channel with a Gaussian input. 
Section IV presents the received signal constellation at the splitting receiver for practical modulation schemes. 
Section V shows the SER results of practical modulation schemes. 
Finally, Section VI concludes the paper.

\underline{Notation:} 
$\tilde{\cdot}$ denotes a complex number. 
$(\cdot)^*$ and $\vert \cdot \vert$ denote the conjugate and the absolute-value norm of a complex number, respectively.
$\mathrm{Real}\{\cdot\}$ and $\mathrm{Imag}\{\cdot\}$ denotes the real part and the imaginary part of a complex number, respectively.
$\myprobability{\cdot}$ denotes the probability of an event. 
$h(\cdot)$, $h(\cdot,\cdot)$, $h(\cdot \vert \cdot)$ denote the differential entropy, joint and conditional differential entropy, respectively. $\mathcal{I}(\cdot;\cdot)$ denotes the mutual information.
Random variables and their realizations are denoted by upper and lower case letters, respectively.
$\mathrm{erfc}(\cdot)$ is the complementary error function, and
$Q(x) \triangleq \frac{1}{2} \mathrm{erfc}(\frac{x}{\sqrt{2}})$ is the Q-function.

\section{System Model}
Consider the communication between a single-antenna transmitter and a $K$-antenna receiver.
The average received signal power at each antenna is denoted by $\mypower$.
The channel coefficient at the $k$th receiver antenna is denoted by $\tilde{h}_k$.

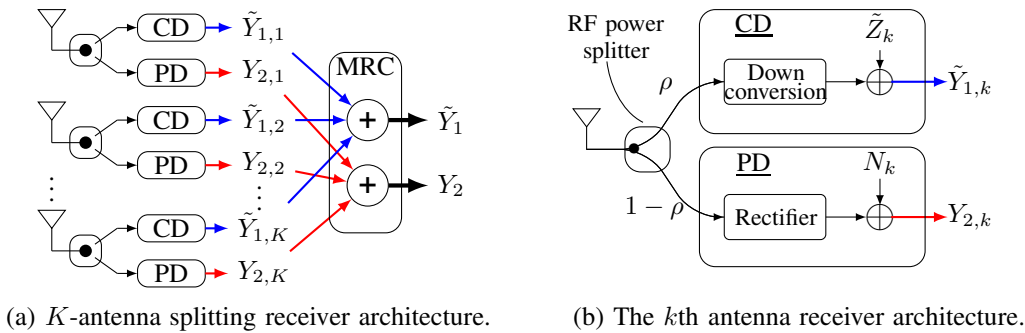
\begin{figure}[t]	
	\small
	\renewcommand{\captionlabeldelim}{ }
	\renewcommand{\captionfont}{\small} \renewcommand{\captionlabelfont}{\small}
	\centering
	\usetikzlibrary{arrows}
	\usetikzlibrary{arrows}
\vspace{-0.5cm}
\begin{tikzpicture} [scale=0.6]
\draw [-latex,rounded corners=10pt] (-4.5,0) -- (-3.15,0) -- (-2.15,1.5) -- (1.7,1.5) node (v1) {};
\draw [-latex,rounded corners=10pt] (-4.5,0) -- (-3.15,0) -- (-2.15,-1.5) -- (1.7,-1.5) node (v2) {};

\draw [-latex,rounded corners=10pt] (-4.5,0) -- (-3.15,0) -- (-2.15,1.5) -- (-1.45,1.5) ;
\draw [-latex,rounded corners=10pt] (-4.5,0) -- (-3.15,0) -- (-2.15,-1.5) -- (-1.45,-1.5) ;

\node [scale=1.5] at (2,1.5) {$\oplus$};
\node [scale=1.5] at (2,-1.5) {$\oplus$};

\draw [-latex,blue,thick](2.25,1.5) -- (3.2,1.5) -- (3.5,1.5) ;
\draw [-latex,red,thick](2.25,-1.5) -- (3.2,-1.5) -- (3.5,-1.5);

\draw [-latex,rounded corners=6pt] (-3.65,0.5) rectangle (-2.65,-0.5);
\draw [-latex](2,2.2) -- (2,1.75);
\draw [-latex](2,-0.65) -- (2,-1.25);

\node at (2,2.7) {$\tilde{Z}_k$};
\node at (2,-0.3) {$N_k$};

\node [align=center,font=\footnotesize] at (-3.85,2.45) { RF power\\[-2mm] splitter};

\draw [rounded corners=2pt,fill=white] (-1.45,2) rectangle (0.8,1);
\draw [rounded corners=2pt,fill=white] (-1.45,-1) rectangle (0.8,-2);
\node [align=center,font=\footnotesize] at (-0.3,1.5) { Down\\[-3mm] conversion};
\node [font=\footnotesize] at (-0.3,-1.5) { Rectifier};
\node at (-3.45,0) {$\bullet$};

\node at (-2.75,1.4) {$\rho$};
\node at (-3,-1.3) {$1-\rho$};
\draw  [rounded corners=6pt](-2,3.1) rectangle (3,0.35);
\draw  [rounded corners=6pt](-2,0.05) rectangle (3,-2.6);
\node at (-0.8,2.75) {\underline{CD}};
\node at (-0.8,-0.35) {\underline{PD}};


\node [circle,draw,scale=0.7,font=\Large] (v11) at (-9.35,0.65) {\bf +};
\node [circle,draw,scale=0.7,font=\Large] (v12) at (-9.35,-0.8) {\bf +};

\draw (-16.35,2.2) node (v7) {} -- (-15.6,2.2) node (v4) {};
\draw (v4) -- (-14.9,2.7) {};
\draw (v4) -- (-14.9,1.7) {};
\draw [-latex](-14.9,2.7) -- (-14.45,2.7);
\draw [-latex](-14.9,1.7) -- (-14.45,1.7);
\draw [rounded corners=4pt] (-14.45,3) rectangle (-12.95,2.4);
\draw [rounded corners=4pt] (-14.45,2) rectangle (-12.95,1.4);
\node at (-13.7,2.7) {CD};
\node at (-13.7,1.7) {PD};
\node [fill,circle,scale=0.5] at (-15.6,2.2) {};
\draw [rounded corners=4pt] (-15.95,2.55) rectangle (-15.25,1.85);
\node (v5) at (-11.65,2.7) {$\tilde{Y}_{1,1}$};
\node (v6) at (-11.65,1.7) {$Y_{2,1}$};
\draw [-latex,blue,thick] (-12.95,2.7) -- (v5);
\draw [-latex,red,thick] (-12.95,1.7) -- (v6);
\draw [-latex,blue,thick](v5) -- (v11);
\draw [-latex,red,thick](v6) -- (v12);

\draw (-16.35,0.15) -- (-15.6,0.15) node (v4) {};
\draw (v4) -- (-14.9,0.65) {};
\draw (v4) -- (-14.9,-0.35) {};
\draw [-latex](-14.9,0.65) -- (-14.45,0.65);
\draw [-latex](-14.9,-0.35) -- (-14.45,-0.35);
\draw [rounded corners=4pt] (-14.45,0.95) rectangle (-12.95,0.35);
\draw [rounded corners=4pt] (-14.45,-0.05) rectangle (-12.95,-0.65);
\node at (-13.7,0.65) {CD};
\node at (-13.7,-0.35) {PD};
\node [fill,circle,scale=0.5] at (-15.6,0.15) {};
\draw [rounded corners=4pt] (-15.95,0.5) rectangle (-15.25,-0.2);
\node (v5) at (-11.65,0.65) {$\tilde{Y}_{1,2}$};
\node (v6) at (-11.65,-0.35) {$Y_{2,2}$};
\draw [-latex,blue,thick] (-12.95,0.65) -- (v5);
\draw [-latex,red,thick] (-12.95,-0.35) -- (v6);
\draw [-latex,blue,thick](v5) -- (v11);
\draw [-latex,red,thick](v6)  -- (v12);

\draw (-16.35,-2.25) -- (-15.6,-2.25) node (v4) {};
\draw (v4) -- (-14.9,-1.75) {};
\draw (v4) -- (-14.9,-2.75) {};
\draw [-latex](-14.9,-1.75) -- (-14.45,-1.75);
\draw [-latex](-14.9,-2.75) -- (-14.45,-2.75);
\draw [rounded corners=4pt] (-14.45,-1.45) rectangle (-12.95,-2.05);
\draw [rounded corners=4pt] (-14.45,-2.45) rectangle (-12.95,-3.05);
\node at (-13.7,-1.75) {CD};
\node at (-13.7,-2.75) {PD};
\node [fill,circle,scale=0.5] at (-15.6,-2.25) {};
\draw [rounded corners=4pt] (-15.95,-1.9) rectangle (-15.25,-2.6);
\node (v5) at (-11.65,-1.75) {$\tilde{Y}_{1,K}$};
\node (v6) at (-11.65,-2.75) {$Y_{2,K}$};
\draw [-latex,blue,thick] (-12.95,-1.75) -- (v5);
\draw [-latex,red,thick] (-12.95,-2.75) -- (v6);
\draw [-latex,blue,thick](v5)-- (v11);
\draw [-latex,red,thick](v6) -- (v12);

\draw (-16.35,2.2) -- (-16.35,2.7) -- (-16.65,3.1) -- (-16.05,3.1) -- (-16.35,2.7);
\draw (-16.35,0.15) -- (-16.35,0.65) -- (-16.65,1.05) -- (-16.05,1.05) -- (-16.35,0.65);
\draw (-16.35,-2.25) -- (-16.35,-1.75) -- (-16.65,-1.35) -- (-16.05,-1.35) -- (-16.35,-1.75); 
\node at (-16.35,-0.7) {$\vdots$};
\node at (-11.75,-0.9) {$\vdots$};

\node (v8) at (-7.5,0.65) {$\tilde{Y}_1$};
\node (v9) at (-7.5,-0.8) {$Y_2$};
\draw [-latex,ultra thick](v11) -- (v8);
\draw [-latex,ultra thick](v12) -- (v9);

\node at (4,1.5) {$\tilde{Y}_{1,k}$};
\node at (4,-1.5) {$Y_{2,k}$};

\draw [rounded corners=6pt] (-10.2,2.15) rectangle (-8.6,-1.85);
\node at (-9.4,1.85) {MRC};
\node at (-12.0,-3.75) {(a) $K$-antenna splitting receiver architecture.};
\node at (0.2,-3.75) {(b) The $k$th antenna receiver architecture.};



\draw (-4.5,0) -- (-4.5,0.5) -- (-4.8,0.9) -- (-4.2,0.9) -- (-4.5,0.5); 

\draw (-4.05,1.95) .. controls (-4,1.5) and (-3.75,0.85) .. (-3.25,0.6);
\end{tikzpicture}
\vspace{-0.5cm}	
	\caption{The proposed splitting receiver architecture.}
	\label{fig:splitting_receiver}	
	\vspace{-0.5cm}
\end{figure}

\subsection{Proposed Receiver Architecture}
\emph{Splitting receiver:} The proposed splitting receiver architecture is illustrated in Figs.~\ref{fig:splitting_receiver}(a) and~(b).
In the first stage, the received signal at each antenna is split into two streams by an ideal \emph{passive RF power splitter}. 
We assume there is no power loss or noise introduced during the splitting process~\cite{split_circuit,datasheet,Xunzhou13}.
{\color{black}One stream is sent to the (conventional) CD circuit and the other to the PD circuit.
	The signals in the CD and PD circuits are first converted to the baseband signals and then sampled and digitized by the analog-to-digital converters (ADCs) accordingly, for further processing.
	Specifically, the rectifier-based PD circuit converts the RF signal into a DC signal with a conversion efficiency $\eta$.}
In the second stage, all the $2K$ streams of signal of the $K$ antennas are jointly used for information detection.\footnote{Although the CD and PD circuits may have different detection sensitivity level in practice~\cite{WanchunICC16}, we assume both the circuits are able to detect arbitrarily small power signal for tractability.}
Note that although we focus on the wireless communication application in this paper, the proposed splitting receiver with single-antenna ($K=1$) is also applicable to cable and fibre-optical communication systems.

\textit{Simplified receiver:} 
We also propose a simplified receiver, as a variant of the splitting receiver, where no power splitters are required.
This is illustrated in Fig.~\ref{fig:mixed_receiver}.
In the simplified receiver, $K_1$ antennas ($1 \leq K_1 < K$) are connected to the CD circuit and the remaining antennas are connected to the PD circuits.
This is illustrated in Fig.~\ref{fig:mixed_receiver}.
We assume that the connections are determined offline, hence do not depend on the instantaneous channel coefficients at each antenna.

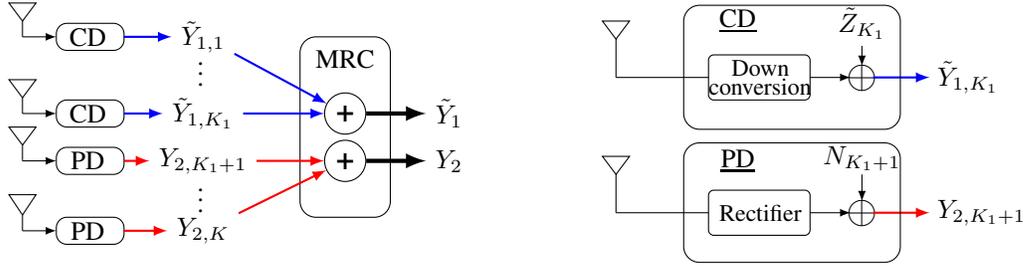
\begin{figure}[t]	
	\small
	\renewcommand{\captionlabeldelim}{ }
	\renewcommand{\captionfont}{\small} \renewcommand{\captionlabelfont}{\small}
	\centering
	\usetikzlibrary{arrows}
	\usetikzlibrary{arrows}
	\vspace{-0.5cm}
	\begin{tikzpicture} [scale=0.6]

\node (v10) at (9.85,2.4) {$\tilde{Y}_{1,1}$};
\draw (5.85,2.4) -- (5.85,2.75) -- (5.55,3.15) -- (6.15,3.15) -- (5.85,2.75);
\draw [-latex] (5.85,2.4) -- (6.6,2.4);
\draw [rounded corners=4pt] (6.6,2.7) rectangle (8.1,2.1);
\node at (7.3,2.4) {CD};
\draw [-latex,blue,thick] (8.1,2.4) -- (v10);

\node (v11) at (9.85,0.7) {$\tilde{Y}_{1,K_1}$};
\draw (5.85,0.7) -- (5.85,1.05) -- (5.55,1.45) -- (6.15,1.45) -- (5.85,1.05);
\draw [-latex] (5.85,0.7) -- (6.6,0.7);
\draw [rounded corners=4pt] (6.6,1) rectangle (8.1,0.4);
\node at (7.3,0.7) {CD};
\draw [-latex,blue,thick] (8.1,0.7) -- (v11);

\node (v12) at (9.85,-0.35) {$Y_{2,K_1 +1}$};
\draw (5.85,-0.35) -- (5.85,0) -- (5.55,0.4) -- (6.15,0.4) -- (5.85,0);
\draw [-latex] (5.85,-0.35) -- (6.6,-0.35);
\draw [rounded corners=4pt] (6.6,-0.05) rectangle (8.1,-0.65);
\node at (7.3,-0.35) {PD};
\draw [-latex,red,thick] (8.1,-0.35) -- (v12);

\node (v13) at (9.85,-1.9) {$Y_{2,K}$};
\draw (5.85,-1.9) -- (5.85,-1.55) -- (5.55,-1.1) -- (6.15,-1.1) -- (5.85,-1.55);
\draw [-latex] (5.85,-1.9) -- (6.6,-1.9);
\draw [rounded corners=4pt] (6.6,-1.6) rectangle (8.1,-2.2);
\node at (7.3,-1.9) {PD};
\draw [-latex,red,thick] (8.1,-1.9) -- (v13);

\node at (9.8,1.75) {$\vdots$};
\node at (9.8,-1.05) {$\vdots$};

\node [circle,draw,scale=0.7,font=\Large] (v111) at (13,0.7) {\bf +};
\node [circle,draw,scale=0.7,font=\Large] (v222) at (13,-0.35) {\bf +};

\draw [rounded corners=6pt] (12,2.4) rectangle (14,-1.6);
\draw [-latex,blue,thick](v10) -- (v111);
\draw [-latex,blue,thick](v11) -- (v111);
\draw [-latex,red,thick](v12) -- (v222);
\draw [-latex,red,thick](v13) -- (v222);
\node at (13,1.9) {MRC};

\node (v8) at (15.3,0.7) {$\tilde{Y}_1$};
\node (v9) at (15.3,-0.35) {$Y_2$};
\draw [-latex,ultra thick](v111) -- (v8);
\draw [-latex,ultra thick](v222) -- (v9);


\draw [-latex] (19,1.5) -- (19.35,1.5) -- (20.35,1.5) -- (24.2,1.5) node (v1) {};
\draw [-latex] (19,-1.5) -- (19.35,-1.5) -- (20.35,-1.5) -- (24.2,-1.5) node (v2) {};

\node [scale=1.5] at (24.45,1.5) {$\oplus$};
\node [scale=1.5] at (24.45,-1.5) {$\oplus$};


\draw [-latex,blue,thick](24.7,1.5) -- (25.65,1.5) -- (25.95,1.5) ;
\draw [-latex,red,thick](24.7,-1.5) -- (25.65,-1.5) -- (25.95,-1.5);

\draw [-latex](24.45,2.2) -- (24.45,1.75);
\draw [-latex](24.45,-0.65) -- (24.45,-1.25);

\node at (24.45,2.7) {$\tilde{Z}_{K_1}$};
\node at (24.45,-0.3) {$N_{K_1+1}$};


\draw [rounded corners=2pt,fill=white] (21.05,2) rectangle (23.3,1);
\draw [rounded corners=2pt,fill=white] (21.05,-1) rectangle (23.3,-2);
\node [align=center,font=\footnotesize] at (22.2,1.5) { Down\\[-3mm] conversion};
\node [font=\footnotesize] at (22.2,-1.5) { Rectifier};

\draw  [rounded corners=6pt](20.5,3.1) rectangle (25.3,0.35);
\draw  [rounded corners=6pt](20.5,0.05) rectangle (25.3,-2.6);
\node at (21.7,2.75) {\underline{CD}};
\node at (21.7,-0.35) {\underline{PD}};

\draw (19,1.5) -- (19,2.35) -- (18.7,2.75) -- (19.3,2.75) -- (19,2.35);
\draw (19,-1.5) -- (19,-0.65) -- (18.7,-0.25) -- (19.3,-0.25) -- (19,-0.65);
\node at (26.8,1.5) {$\tilde{Y}_{1,K_1}$};
\node at (27.1,-1.5) {$Y_{2,K_1+1}$};


\node at (9.9,-3.5) {(a) $K$-antenna simplified receiver architecture.};
\node at (24,-3.5) {(b) The $K_1$th and the $(K_1+1)$th antenna receiver architectures.};
	\end{tikzpicture}
	\vspace{-0.5cm}
	\caption{The proposed simplified receiver architecture.}
	\label{fig:mixed_receiver}	
	\vspace{-0.5cm}
\end{figure}

\subsection{Signal Model}
In this section, we present the signal model for the splitting receiver.
Note that the simplified receiver can be analytically treated as a special case of the splitting receiver with power splitting ratios taking binary values only, i.e., $\rho_k \in \{0,1\}$, for all $k=1,2,..., K$.

Based on~\cite{Xunzhou13,OpticalTIT}, the output signals from the CD and PD circuits at the $k$th antenna are given by, respectively,
\begin{align}
\tilde{Y}_{1,k}& = \sqrt{\rho_k \myP} \tilde{h}_k \tilde{X} + \tilde{Z}_k, \label{receive_signal_1}\\
Y_{2,k}&=  \eta (1-\rho_k)  \vert \tilde{h}_k \vert^2 \myP \vert \tilde{X} \vert^2 + N'_k,  \label{receive_signal_2'}
\end{align}
where $\rho_k \in \left[0,1\right]$ is the power splitting ratio.
$\tilde{X}$ is the transmitted signal with normalized variance and $\tilde{X} \in \mathcal{X}$, where $\mathcal{X}$ denotes the set of all possible transmitted signals.
{\color{black}$\tilde{Z}_k$ is the post-processing complex AWGN of the CD circuit with the mean of zero and the variance of $\sigmaone$, which includes both the RF band to baseband conversion noise and the ADC noise. 	
	$N'_k$ is the post-processing noise of the PD circuit which is also assumed to be real Gaussian noise~\cite{OpticalTIT}, which includes both the rectifier noise and the ADC noise.}
Note that we only consider the post-processing noise $\tilde{Z}_k$ and $N'_k$, i.e., we ignore the pre-processing noise, such as the antenna noise which is almost at the thermal noise level and is much smaller than the post-processing noise~\cite{Xunzhou13}.

Without loss of generality, scaling \eqref{receive_signal_2'} by $\eta $, the received signal {\color{black}$Y_{2,k}$} can be rewritten~as
\begin{align}
Y_{2,k} &= (1-\rho_k) \vert \tilde{h}_k \vert^2  \myP \vert \tilde{X} \vert^2 + {\color{black}N_k} \label{receive_signal_2},
\end{align}
where $N_k \triangleq N'_k/(\eta )$ is the equivalent rectifier conversion AWGN with the mean of zero and the variance $\sigmatwo$.
%


\subsection{Maximal Ratio Combining of Splitting Receiver}
To detect the transmitted signal $\tilde{X}$, similar with a conventional SIMO receiver, the optimal method is the maximal ratio combining (MRC).
We assume that the receiver has perfect channel state information (CSI), i.e., knowledge of $\tilde{h}_k$.
Since the $K$-antenna received signals $\tilde{Y}_{1,k}$ and $Y_{2,k}$, $k=1,2,...,K$, lie in different signal spaces, we use MRC for coherently processed signals (i.e., $\tilde{Y}_{1,k}$) and non-coherent signals (i.e., $Y_{2,k}$) separately.
Based on~\eqref{receive_signal_1} and~\eqref{receive_signal_2}, the combined coherently and non-coherently processed signals are given by, respectively,
\begin{equation} \label{MRC_1}
\begin{aligned}
\tilde{Y}_1 &=  \left(\sum\limits_{k=1}^{K} \rho_k \vert \tilde{h}_k \vert^2 \right) \sqrt{\myP} \tilde{X} + \sum\limits_{k=1}^{K} \sqrt{\rho_k}  \tilde{h}^*_k  \tilde{Z}_k, \\
Y_2 &= \left(\sum\limits_{k=1}^{K} (1-\rho_k)^2 \vert \tilde{h}_k \vert^4  \right) \myP \vert \tilde{X} \vert^2 +  \sum\limits_{k=1}^{K} (1-\rho_k) \vert \tilde{h}_k \vert^2 N_k.
\end{aligned}
\end{equation}
For convenience of analysis, after linear scaling, \eqref{MRC_1} can be rewritten as
\begin{equation} \label{MRC_signal}
\begin{aligned}
\tilde{Y}_1 &=  \sqrt{\Theta_1} \sqrt{\myP} \tilde{X} + \tilde{Z}, \ 
Y_2 = \sqrt{\Theta_2} \myP \vert \tilde{X} \vert^2 +  N,
\end{aligned}
\end{equation}
where
\begin{equation} \label{my_theta}
\begin{aligned}
\Theta_1 &=  {\sum\limits_{k=1}^{K} \rho_k \vert \tilde{h}_k \vert^2 } , \ 
\Theta_2 = {\sum\limits_{k=1}^{K} (1-\rho_k)^2 \vert \tilde{h}_k \vert^4},
\end{aligned}
\end{equation}
and $\tilde{Z}$ and $N$ follow the same distributions as $\tilde{Z}_k$ and $N_k$, respectively.
The two-dimensional signal $\tilde{Y}_1$ and the one-dimensional signal $Y_2$ form a triple $(\tilde{Y}_1, Y_2)$, which is the equivalent received signal of the $K$-antenna splitting receiver.

\textnormal{It is interesting to see that since the two-dimensional signal $\tilde{Y}_1$ lies on the \emph{in-phase-quadrature} (I-Q) plane and the one-dimensional signal $Y_2$ lines on the \emph{power} (P)-axis, the equivalent received signal $(\tilde{Y}_1, Y_2)$ lies in the three-dimensional I-Q-P space. 
This is different from the conventional coherent (two-dimensional) and non-coherent (one-dimensional) receiver signal spaces.
Thus, the splitting receiver expands the received signal space and fundamentally changes the way in which the signal is processed compared with the conventional receivers.}


Considering the noiseless signal, i.e., letting $\tilde{Z}$ and $N = 0$ in \eqref{MRC_signal}, we have $Y_2 = \frac{\sqrt{\Theta_2}}{\Theta_1}\vert\tilde{Y}_1\vert^2$ from \eqref{MRC_signal}, which is a paraboloid equation. 
From a geometric point of view, defining $\vec{\rho} \triangleq [\rho_1, \rho_2,...\ \rho_K]$, $\vec{1} \triangleq \underbrace{[1, 1, \cdots, 1]}_K$, and $\vec{0} \triangleq \underbrace{[0, 0, \cdots, 0]}_{K}$,
the splitting receiver is actually bending the noiseless received signal space into a paraboloid with $\vec{\rho}$ as illustrated in Fig.~\ref{fig:Gaussian_shape}. 
When $\vec{\rho} = \vec{1}$, i.e., a non-splitting case, the splitting receiver degrades to the coherent receiver.
As the parameter $\sqrt{\Theta_2}/\Theta_1$ increases, e.g., each element of $\vec{\rho}$ decreases, the splitting receiver bends the signal plane to a paraboloid, which is taller and thinner with a larger $\sqrt{\Theta_2}/\Theta_1$ , e.g., a smaller $\vec{\rho}$.
When $\vec{\rho} = \vec{0}$, the splitting receiver degrades to the PD-based non-coherent receiver.

In this paper, PD-based non-coherent receiver is named as \emph{non-coherent receiver} for short, and we refer to both the $K$-antenna coherent receiver (i.e., $\vec{\rho} = \vec{1}$) and the $K$-antenna non-coherent receiver (i.e., $\vec{\rho} = \vec{0}$) as the \emph{conventional receivers}.

\begin{figure}[t]
	\small
	\renewcommand{\captionlabeldelim}{ }
	\renewcommand{\captionfont}{\small} \renewcommand{\captionlabelfont}{\small}	
	\centering 
	\vspace{-1.3cm}
	\includegraphics[scale=0.8]{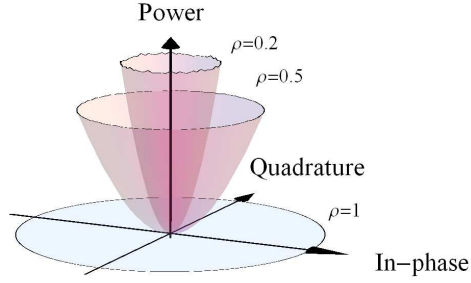}	
	\vspace*{-0.5cm}
	\caption{\small Illustration of the signal space of the splitting receiver, $\rho= \ 0.2,\ 0.5,\ 1$.}	
	\label{fig:Gaussian_shape}
	\vspace*{-0.5cm}
\end{figure}
\subsection{Splitting Channel}
From an information theory perspective, \eqref{MRC_signal} can be rewritten as
\begin{equation} \label{receive_signal}
\begin{bmatrix} \tilde{Y}_1 \\ Y_2 \end{bmatrix} 
= \begin{bmatrix} 1 & 0 \\ 0 & \vert \cdot \vert^2 \end{bmatrix} 
\begin{bmatrix} \sqrt{\Theta_1} \\  \sqrt[4]{\Theta_2} \end{bmatrix} \sqrt{\myP}\tilde{X} 
+ \begin{bmatrix} \tilde{Z} \\ N \end{bmatrix},
\end{equation}
where $\vert \cdot \vert^2$ is the squared magnitude operator. 
We name \eqref{receive_signal} as the \emph{splitting channel},
and the input and output of the splitting channel regarded as random variables, are $\sqrt{\myP} \tilde{X}$ and $(\tilde{Y}_1,Y_2)$, respectively.

The splitting channel can be treated as a SIMO channel, since the channel has one input $\sqrt{\myP}\tilde{X}$ and two outputs $\tilde{Y}_1$ and $Y_2$.
It can also be treated as a degraded (due to the power splitting) SISO channel with the output $\tilde{Y}_1$ and a side information $Y_2$.

\subsection{Performance Metrics}
We study the mutual information between the input and output of the splitting channel for an ideal Gaussian input, and study the SER performance for practical modulation schemes.

For convenience of analysis, we define the operating SNR as 
\begin{equation}
\snr \triangleq \min\left\lbrace \snrcov, \snrrec  \right\rbrace,
\end{equation}
where 
\begin{equation}
\snrcov \triangleq H_2 \frac{\myP}{\sigmaone},\ 
\snrrec \triangleq \sqrt{H_4} \frac{\myP}{\sigmatwosqrt},\ 
H_2 \triangleq \sum\limits_{k=1}^{K} \vert \tilde{h}_k \vert^2,\ 
H_4 \triangleq \sum\limits_{k=1}^{K} \vert \tilde{h}_k \vert^4.
\end{equation}
$\snrcov$ and $\snrrec$ are the SNRs of the conventional receivers, i.e., $\vec{\rho} =\vec{1}$ for the coherent receiver and $\vec{\rho} = \vec{0}$ for the non-coherent receiver, respectively.
Specifically, the definition of $\snrrec$ is consistent with \cite{OpticalTIT}. 
Note that although $\sqrt{H_4} \myP$ and $\sigmatwosqrt$ correspond to the standard deviation (not variance) of the signal $\sqrt{H_4} \myP \vert\tilde{X} \vert^2$ and the noise $N$ at the PD receiver, respectively, $\sqrt{H_4} \myP$ still has the physical meaning of ``power''. Thus, the signal-to-noise ratio is defined as $\sqrt{H_4} \frac{\myP}{\sigmatwosqrt}$ not $H_4\frac{\myP^2}{\sigmatwo}$.

In the following, 
we refer to the high SNR regime as $\snr \rightarrow \infty$ which is obtained by letting $\myP\rightarrow \infty$.
Our analysis will focus on the splitting receiver which includes the simplified receiver as a special case.

\section{Splitting Channel: Mutual Information}
In this section, we study the mutual information of the splitting channel to determine the gain due to the joint coherent and non-coherent processing.
We also provide a discussion to intuitively explain this processing gain.

Based on \eqref{receive_signal}, the mutual information between the input and outputs of the splitting channel with the splitting ratio $\vec{\rho}$ is 
\begin{equation} \label{first_mutual_info}
\begin{aligned}
&\mathcal{I}(\sqrt{\myP}\tilde{X};\tilde{Y}_1,Y_2) = h(\tilde{Y}_1,Y_2) - h(\tilde{Y}_1,Y_2 \vert \sqrt{\myP} \tilde{X})\\
&{=} h(\tilde{Y}_1,Y_2) - h(\tilde{Z}, N \vert \sqrt{\myP} \tilde{X})
{=} h(\tilde{Y}_1,Y_2) - h(\tilde{Z}, N)
{=} h(\tilde{Y}_1,Y_2) - h(\tilde{Z}) - h(N)\\
&= - \int_{Y_2} \int_{\tilde{Y}_1} f_{\tilde{Y}_1, Y_2} (\tilde{y}_1,y_2) \log_2(f_{\tilde{Y}_1, Y_2} (\tilde{y}_1,y_2)) \mathrm{d}\tilde{y}_1 \mathrm{d}y_2 - \log_2(\pi e \sigmaone) - \frac{1}{2} \log_2(2\pi e \sigmatwo).
\end{aligned}
\end{equation}
The joint probability density function (pdf) of $(\tilde{Y}_1, Y_2)$ is
\begin{equation}
f_{\tilde{Y}_1, Y_2} (\tilde{y}_1,y_2) = \int_{\tilde{X}} f_1(\sqrt{\Theta_1 \myP} \tilde{x},\tilde{y}_1)  f_2(\sqrt{\Theta_2} \myP \vert \tilde{x} \vert^2,y_2) f_{\tilde{X}} (\tilde{x})  \mathrm{d} \tilde{x},
\end{equation}
where  
$f_{\tilde{X}}(\tilde{x})$ is the pdf of $\tilde{X}$, and 
$f_1(\sqrt{\Theta_1 \myP} \tilde{x},\cdot)$ and $f_2(\sqrt{\Theta_2} \myP \vert \tilde{x} \vert^2, \cdot)$ are the pdfs of the distributions $\mathcal{CN}(\sqrt{\Theta_1 \myP} \tilde{x},\sigmaone)$ and $\mathcal{N}(\sqrt{\Theta_2}\myP \vert \tilde{x} \vert^2,\sigmatwo)$, respectively.

The mutual information expression in~\eqref{first_mutual_info} needs five integrals to evaluate, which is cumbersome and thus the maximal mutual information theoretically achieved by the optimal distribution of $\tilde{X}$, cannot be obtained.
%

\subsection{Mutual Information and Joint Processing Gain}
For tractability, in the following analysis, we consider the mutual information with a Gaussian input, i.e., $\tilde{X} \sim \mathcal{CN}(0,1)$, and we have:
\begin{enumerate} [1)]
	\item Letting $\vec{\rho} = \vec{1}$, the splitting channel is degraded to the coherent AWGN channel, and the mutual information is well-known~as~\cite{BookInfo}
	\begin{equation} \label{rho_1}
	 \mathcal{I}(\sqrt{\myP}\tilde{X};\tilde{Y}_1,Y_2) 
	 = h(\tilde{Y}_1) -h(\tilde{Z})
	 = \log_2 \left(1 + H_2\frac{\myP}{\sigmaone}\right),
	\end{equation}
	which is exactly the capacity of the coherent AWGN channel, i.e., $\mathcal{C}(\vec{\rho} = \vec{1})$.
	\item Letting $\vec{\rho} = 0$, the splitting channel is degraded to the conventional intensity channel in free-space optical communications~\cite{OpticalTIT}. Recall that we refer to the intensity channel as the non-coherent AWGN channel, echoing the coherent AWGN channel in this paper\footnote{Note that in this paper, the non-coherent channel refers to the intensity channel, and it does not refer to the kind of channel without CSI at the transmitter or the receiver.}. The mutual information of the non-coherent AWGN channel is~\cite{OpticalTIT}
	\begin{align}
		 \mathcal{I}(\sqrt{\myP}\tilde{X};\tilde{Y}_1,Y_2)
		 &= h(Y_2) -h(N)
		 = -\int\limits_{-\infty}^{\infty} f_{Y_2}(y_2) \log_2(f_{Y_2}(y_2)) \mathrm{d}y_2 - \frac{1}{2} \log_2(2 \pi e \sigmatwo) \nonumber\\
		 &\stackrel{(a)}{\geq} \frac{1}{2} \log_2\left(1 + H_4 \frac{\myP^2 e}{2 \pi \sigmatwo}\right),\label{rho_0}
	\end{align}
	where $Y_2 = \sqrt{H_4} \myP \vert \tilde{X}\vert^2 + N$ follows an exponential modified Gaussian distribution~\cite{EMG}: 
	\begin{equation}
	 f_{Y_2}(y_2) = \frac{1}{2 \sqrt{H_4} \myP} \exp\left(\frac{1}{2 \sqrt{H_4} \myP}\left(\frac{\sigmatwo}{\sqrt{H_4} \myP} -2 y_2\right)\right) \mathrm{erfc}\left(\frac{\frac{\sigmatwo}{\sqrt{H_4} \myP} - y_2}{\sqrt{2}\sigmatwosqrt}\right).
	\end{equation}
	The inequality $(a)$ is given by~\cite{OpticalTIT}, and \eqref{rho_0} is the asymptotic mutual information in the high SNR regime,
	which is also the asymptotic capacity (with gap less than $\vert \frac{1}{2} \log_2\left(\frac{e}{2 \pi}\right) \vert$ bits) of the non-coherent AWGN channel, i.e., $\mathcal{C}(\vec{\rho} = \vec{0})$.
\end{enumerate}

Comparing \eqref{rho_1} and \eqref{rho_0}, it is easy to see that as $\snr \rightarrow \infty$, 
\emph{the coherent and non-coherent AWGN channels have the same asymptotic capacity}, i.e., $\lim_{\snr \rightarrow \infty} {\mathcal{C}(\vec{\rho}=\vec{1})}/{\mathcal{C}(\vec{\rho}=\vec{0})} = 1$.
%
In the following, we will show that the splitting receiver with $\vec{\rho} \neq \vec{0} \text{ nor } \vec{1}$ provides a gain in the mutual information compared with the conventional receivers.
Firstly, 
we need the following definition.
\begin{definition} \label{def:splitting_gain_MI}
The joint processing gain of the splitting receiver is 
	\begin{equation}
	G \triangleq  \frac{\sup\{\mathcal{I}(\sqrt{\myP}\tilde{X};\tilde{Y}_1,Y_2): \vec{\rho} \in \left[0,1\right]^K\}}{\max\{\mathcal{I}(\sqrt{\myP}\tilde{X};\tilde{Y}_1,Y_2)\vert_{\vec{\rho} = \vec{0}},\mathcal{I}(\sqrt{\myP}\tilde{X};\tilde{Y}_1,Y_2)\vert_{\vec{\rho} = \vec{1}}\}},
	\end{equation}
	where $\sup\{\cdot\}$ denotes for the supremum, and $\left[0,1\right]^K$ is the $K$-product space generated by the interval $\left[0,1\right]$.
\end{definition}

If the joint processing gain $G > 1$, the splitting receiver achieves higher mutual information compared with the conventional receivers.
If the joint processing gain $G = 1$ which means the joint coherent and non-coherent processing is unnecessary, the splitting receiver should be degraded to either one of the conventional receivers.

Due to the complicated form of \eqref{first_mutual_info},
it is not possible to accurately evaluate the mutual information\footnote{A lower bound and an upper bound of $\mathcal{I}(\sqrt{\myP}\tilde{X};\tilde{Y}_1, Y_2)$ with explicit expressions can be found based on the basic inequalities $\mathcal{I}(\sqrt{\myP}\tilde{X};\tilde{Y}_1, Y_2) >\mathcal{I}(\sqrt{\myP}\tilde{X};\tilde{Y}_1)$, $\mathcal{I}(\sqrt{\myP}\tilde{X};\tilde{Y}_1, Y_2) >\mathcal{I}(\sqrt{\myP}\tilde{X};Y_2)$ and $\mathcal{I}(\sqrt{\myP}\tilde{X};\tilde{Y}_1, Y_2) < \mathcal{I}(\sqrt{\myP}\tilde{X};\tilde{Y}_1) +\mathcal{I}(\sqrt{\myP}\tilde{X};Y_2)$~\cite{BookInfo}.
Since the bounds are loose, we do not pursue them here.} for $\vec{\rho} \in \left[0,1 \right]^K$ and prove whether $G$ is greater than $1$ or not. 
Hence, we first use Monte Carlo based histogram method to simulate the results.
In Fig.~\ref{fig:mutual_new}, considering the $K=1$ case, it is observed that when $\snr$ is reasonably high, e.g., $\myP = 10$, $\sigmaone = 1$ and $\sigmatwosqrt =1$, the joint processing gain $G$ is greater than $1$. 
Inspired by this, we will focus on the analysis on the mutual information in \eqref{first_mutual_info} and the joint processing gain in Definition~\ref{def:splitting_gain_MI} in the high SNR regime in the following subsection.

\begin{figure}[t]
	\small
	\renewcommand{\captionlabeldelim}{ }
	\renewcommand{\captionfont}{\small} \renewcommand{\captionlabelfont}{\small}	
	\centering 
	\vspace*{-0.7cm}
	\includegraphics[scale=0.6]{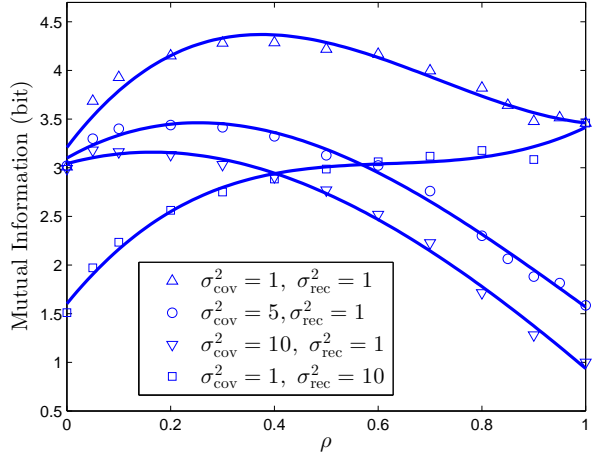}	
	\vspace*{-0.5cm}
	\caption{\small $\mathcal{I}(\sqrt{\myP}\tilde{X};\tilde{Y}_1,Y_2)$ versus $\rho$, $\myP = 10$, $K=1$, $\vert \tilde{h}_1 \vert = 1$. The simulation results are marked with `o's, and are curve fitted by polynomials of degree of $3$.}	
	\label{fig:mutual_new}
	\vspace*{-0.5cm}
\end{figure}

\subsection{High SNR Analysis}
\begin{lemma} \label{theory:high_snr}
In the high SNR regime, 
$\mathcal{I}(\sqrt{\myP}\tilde{X};\tilde{Y}_1,Y_2)$ with $\vec{\rho} \in \left[0,1 \right]^K \backslash\{\vec{0},\vec{1}\}$ is given by
\begin{subequations}
\begin{alignat}{2}
 \label{theory1_1}
	\mathcal{I}(\sqrt{\myP}\tilde{X};\tilde{Y}_1,Y_2) 
&\approx \log_2( \frac{\Theta_1 \myP}{\sigmaone}) +\frac{1}{2 \log(2)} \exp\left(\frac{\Theta_1 \sigmatwo }{\Theta_2 2 \sigmaone \myP }\right) \Ei \left(\frac{\Theta_1 \sigmatwo}{\Theta_2 2 \sigmaone \myP }\right)\\
 \label{theory1_2}
&\approx \log_2(\frac{\sqrt{2} \myP^{\frac{3}{2}} \sqrt{\Theta_1 \Theta_2}}{\sigmaonesqrt \sigmatwosqrt}) - \frac{\gamma}{2 \ln 2},
\end{alignat}
\end{subequations}
where $\Ei(x) \triangleq \int_{x}^{\infty} \frac{e^{-t}}{t} \mathrm{d} t$ is the exponential integral function, and  $\gamma$ is Euler's constant.
\end{lemma}
\begin{proof}
	See Appendix A.
\end{proof}
{\color{black}From Lemma~1, it is clear that the mutual information of the splitting channel increases linearly with $\log_2(\myP)$ and decreases linearly with $\log_2(\sigmaonesqrt)$ and  $\log_2(\sigmatwosqrt)$ in the high SNR regime. 
Moreover, since the mutual information depends on the power splitting ratio $\vec{\rho}$, which is contained in the term $\Theta_1 \Theta_2$, it is interesting to find the optimal $\vec{\rho}$ that maximizes the mutual information.}

Based on Lemma~\ref{theory:high_snr}, the following optimization problem is proposed to obtain the optimal splitting ratio $\vec{\rho}$ in the high SNR regime:
\begin{equation} \label{first_P1}
\textrm{(P1)} \ \max_{\vec{\rho} \in \left[0,1 \right]^K\backslash\{\vec{0},\vec{1}\}}  \Theta_1  \Theta_2 
\Leftrightarrow  \max_{\vec{\rho} \in \left[0,1 \right]^K\backslash\{\vec{0},\vec{1}\}} 
\sum\limits_{k=1}^{K} \rho_k \vert \tilde{h}_k \vert^2
\sum\limits_{k=1}^{K} \left(1-\rho_k\right)^2 \vert \tilde{h}_k \vert^4.
\end{equation}
It can be shown  that (P1) is not a convex optimization problem. Thus, the optimal splitting ratio can be obtained by numerical methods.
In what follows, we first focus on two special scenarios and then discuss the joint processing gain for a general splitting receiver.\par
\subsubsection{Splitting receiver with single receiver antenna}
When $K=1$, $\Theta_1 = \rho_1 \vert \tilde{h}_1 \vert^2$ and $\Theta_2 = (1-\rho_1)^2 \vert \tilde{h}_1 \vert^4$,
and thus, 
the optimal power-splitting ratio is obtained by solving the equation
$\frac{\partial \rho^{\frac{1}{3}} (1-\rho)^{\frac{2}{3}}}{\partial \rho} = 0.$
It is straightforward to obtain the following results.
\begin{proposition} \label{coro:opti_rho}
	For the splitting receiver with single receiver antenna, the optimal splitting ratio in the high SNR regime is
	\begin{equation}
	\rho^\star = \frac{1}{3},
	\end{equation}
	and the maximal mutual information is given by
	\begin{equation} \label{optimal_MI}
	\begin{aligned}
	\mathcal{I}(\sqrt{\myP}\tilde{X};\tilde{Y}_1,Y_2) \vert_{\rho^\star}
	\approx \log_2 \left(\frac{2 \sqrt{2}}{3\sqrt{3}} \frac{\vert \tilde{h}_1 \vert^3 \myP^{\frac{3}{2}}}{\sigmaonesqrt \sigmatwosqrt}\right)- \frac{\gamma}{2 \ln 2}.
	\end{aligned}
	\end{equation} 
\end{proposition}

\subsubsection{Simplified receiver with a large number of antennas} 
For the simplified receiver with a large number of antennas, \eqref{first_P1} can be rewritten as 
\begin{equation}
\max_{1 \leq K_1 <K} 
\sum\limits_{k =1}^{K_1}\vert \tilde{h}_k \vert^2
\sum\limits_{k = K_1+1}^{K}  \vert \tilde{h}_k \vert^4.
\end{equation}
Assuming that $\tilde{h}_k$, $k =1,2,... K$, are independent and identically distributed  (i.i.d.) random variables, i.e., \emph{\color{black}the uncorrelated scenario}, due to the law of large numbers when $K$ is sufficiently large, we have
\begin{equation} \label{antennas}
\lim\limits_{K \rightarrow \infty} \sum\limits_{k = 1}^{K_1}\vert \tilde{h}_k \vert^2
\!\!\!\sum\limits_{k = K_1+1}^{K} \!\!\!\! \vert \tilde{h}_k \vert^4
= \lim\limits_{K \rightarrow \infty} K_1 K_2 {\sum\limits_{k = 1}^{K_1}\vert \tilde{h}_k \vert^2}/{K_1} 
\!\!\!\sum\limits_{k = K_1+1}^{K} \!\!\!\! \vert \tilde{h}_k \vert^4/{K_2}
\stackrel{(a)}{=} K_1 K_2 \myexpect{\vert \tilde{h}_k \vert^2}\myexpect{\vert \tilde{h}_k \vert^4},
\end{equation}
where $K_2 \triangleq K-K_1$, and $(a)$ is because both $K_1$ and $K_2$ are sufficiently large.
{\color{black}Assuming that $\vert \tilde{h}_k \vert $, $k =1,2,... K$, are identical with each other, i.e., \emph{the free-space scenario which is also a fully-spatially-correlated scenario}, we have the same expression with~\eqref{antennas}.}
Thus, $K_1 K_2$ is maximized when $K_1 = K_2 = K/2$, and we have the following proposition.
\begin{proposition} \label{prop:large_antenna}
For the simplified receiver with a large number of antennas, {\color{black}the optimal strategy for the spatially-uncorrelated channel or the fully-spatially-correlated channel (i.e., the free-space scenario)} in the high SNR regime is to connect half of the antennas to the CD circuits and the other half to the PD circuits, and the maximum mutual information is given by
\begin{equation}
\mathcal{I}(\sqrt{\myP}\tilde{X};\tilde{Y}_1,Y_2) \vert_{\vec{\rho}^\star}
	\approx 
 \log_2\left(\frac{K \myP^{\frac{3}{2}} \sqrt{ \myexpect{\vert \tilde{h}_k \vert^2}\myexpect{\vert \tilde{h}_k \vert^4}}}{\sqrt{2} \sigmaonesqrt \sigmatwosqrt}\right) - \frac{\gamma}{2 \ln 2}.
\end{equation}
\end{proposition}
{\color{black}Note that for the general spatially-correlated scenario, the optimal strategy in the high SNR regime is not immediately clear. We will investigate this scenario for future study.}


\subsubsection{Joint processing gain of splitting receiver with $K$ receiver antennas}
We assume that $\vec{\rho} \neq \vec{0} \text{ nor } \vec{1}$, thus, $\Theta_1 \neq 0$ and $\Theta_2 \neq 0$.
Then, based on \eqref{theory1_2} of Lemma~\ref{theory:high_snr}, \eqref{rho_1} and \eqref{rho_0}, we can show that 
\begin{equation}
\lim\limits_{\snr \rightarrow \infty} \frac{\mathcal{I}(\sqrt{\myP}\tilde{X};\tilde{Y}_1,Y_2)\vert_{\vec{\rho} \in \left[0,1\right]^K \backslash \{\vec{0},\vec{1}\}}}{\max\{\mathcal{I}(\sqrt{\myP}\tilde{X};\tilde{Y}_1,Y_2)\vert_{\vec{\rho} = \vec{0}},\mathcal{I}(\sqrt{\myP}\tilde{X};\tilde{Y}_1,Y_2)\vert_{\vec{\rho} = \vec{1}}\}} = \frac{3}{2}.
\end{equation}
In other words, 
the asymptotic gain is the same no matter what value $\vec{\rho}$ takes, as long as $\vec{\rho} \neq \vec{0} \text{ nor } \vec{1}$. 
Therefore, the asymptotic optimal splitting ratio $\vec{\rho}^\star \in \left[0,1\right]^K \backslash \{\vec{0},\vec{1}\}$, and we have the following result based on Definition~\ref{def:splitting_gain_MI}.
%
%

\begin{proposition} \label{splitting_gain}
In the high SNR regime, the asymptotic joint processing gain for a splitting receiver with $K$ receiver antennas is
\begin{equation}
G = \lim\limits_{\snr \rightarrow \infty} \frac{\mathcal{I}(\sqrt{\myP}\tilde{X};\tilde{Y}_1,Y_2)\vert_{\vec{\rho}^{\star}}}{\max\{\mathcal{I}(\sqrt{\myP}\tilde{X};\tilde{Y}_1,Y_2)\vert_{\vec{\rho} = \vec{0}},\mathcal{I}(\sqrt{\myP}\tilde{X};\tilde{Y}_1,Y_2)\vert_{\vec{\rho} = \vec{1}}\}} = \frac{3}{2}.
\end{equation}
\end{proposition}

\subsection{Explanation of the Joint Processing Gain}
The result of Proposition~\ref{splitting_gain} shows that in the high SNR regime, since $G=3/2 > 1$, the splitting receiver provides a processing gain.
Note that although the joint processing gain at any given SNR depends on the specific value of the received signal power $\myP$, the noise variances $\sigmaone$ and $\sigmatwo$,
the asymptotic joint processing gain is independent of the specific noise variances at the CD and PD circuits in the high SNR regime.
\emph{This implies that 
the reason for the performance improvement lies in the joint coherent and non-coherent processing.}
This is explained in detail using intuitive and geometric arguments as follows.
%

\emph{Intuitive explanation of the rate improvement:}
\textnormal{Since the degrees of freedom of a channel is commonly defined as the dimension of the received signal space~\cite{BOOKTse}, the coherent AWGN channel has two degrees of freedom (I-Q plane) while the non-coherent AWGN channel has one degree of freedom (P-axis). For the splitting channel created by jointly utilizing both the coherent and non-coherent AWGN channels, the received signals are spread into a three-dimensional space, i.e., the I-Q-P space. 
	Thus, the splitting channel can be treated as a channel with three degrees of freedom. 
	Therefore, the splitting channel with a properly designed splitting ratio can take better advantage of the I-Q-P space, and achieve a better channel rate performance compared with either the coherent or non-coherent AWGN channel.}
	
\textnormal{We would like to highlight that a `splitting receiver', which splits the received signal at each antenna into two streams and sends both streams to CD circuits (i.e., two coherent AWGN channels), does not provide any rate improvement. After MRC, it is straightforward to see that the received signal space still lies on the I-Q plane. Thus, the received signal space is the same as for the conventional coherent receiver.
For instance, consider a single-antenna receiver for ease of illustration. It can be proved that the best `splitting' strategy is to send the entire signal to the CD circuit with~a smaller noise variance, instead of splitting and sending signals to both CD circuits~\cite{BOOKTse}.
Therefore, there is no joint processing gain by using two coherent AWGN channels, i.e., $G = 1$. 
The same argument holds for a `splitting receiver' which splits the received signal at each antenna into two streams and sends them to two PD circuits (i.e., two non-coherent AWGN channels).
}

\textnormal{Therefore, the key to the rate improvement is the increased dimension of the received signal space achieved by joint coherent and non-coherent processing, where the coherent channel adds noise linearly to the signal, and the noncoherent channel adds noise to the square amplitude of the signal.	
	}

\emph{A geometric explanation of the asymptotic gain:}
\textnormal{As discussed in Sec.~II.C, a splitting receiver with the splitting ratio $\vec{\rho}$ maps the noiseless received signal space, i.e., the I-Q plane, to a paraboloid in the I-Q-P space with parameter ${\sqrt{\Theta_2}}/{\Theta_1}$ which depends on $\vec{\rho}$. 
Considering a disk with radius $R$ and center $(0,0)$ in the I-Q plane, the area of the disk is $\pi R^2$, {\color{black}where $R$ is proportional to $\sqrt{P}$ in this paper.} 
After the mapping, the disk is converted into a paraboloid with parameter ${\sqrt{\Theta_2}}/{\Theta_1}$ which is restricted by the condition that the projection of the paraboloid in the I-Q plane should lie within the disk with radius $\sqrt{\Theta_1}R$.
When $R$ is sufficiently large, the area of the paraboloid can be shown to be approximated by $3 \pi \sqrt{\Theta_1 \Theta_2} R^3$ for $\vec{\rho} \neq \vec{0} \text{ nor } \vec{1}$.
It is well known that the optimal constellation design for the I-Q space is equivalent to a sphere-packing problem, i.e., packing two-dimensional spheres (disks) with a certain radius, which is related to the detection error rate, on the surface of the disk (i.e., the disk on the I-Q plane). 
The number of spheres that can be packed is proportional to the area of the disk.
Thus, the communication rate can be written as $\mathcal{O}\left(\log({\pi R^2})\right)\sim 2\mathcal{O}({\color{black}\log{R}})$.
Similarly, for the paraboloid, the number of three-dimensional spheres\footnote{Note that sphere-packing is considered only if  $\sigmaone = 2 \sigmatwo$, i.e., a uniform three-dimensional noise sphere, otherwise, it is ellipsoid-packing. Here we use sphere-packing for ease of illustration.} (balls) that can be packed on the surface is proportional to the paraboloid area, and the rate can be written as $\mathcal{O}\left(\log({3 \pi \sqrt{\Theta_1 \Theta_2} R^3})\right)\sim 3\mathcal{O}({\color{black}\log{R}})$.
Therefore, it is straightforward to see that there is a $3/2$ fold rate gain provided by the splitting receiver when $R$ is sufficiently large.
To sum up, bending the signal space from a two-dimensional plane to a three-dimensional paraboloid increases the effective area of the signal space which boosts the communication rate.}

%
%

\emph{The complexity of splitting receiver:}
{\color{black}Although the splitting receiver is able to provide a performance gain, it is clear that for the information detection in the digital domain, the splitting receiver requires a three-dimensional detection, while the conventional CD/PD receiver only needs a two/one-dimensional detection, respectively. 
	Specifically, when applying the minimum distance detection for practical modulation, the splitting receiver needs to calculate the distance between two signal points in the three-dimensional space, while the conventional CD/PD receiver only needs to calculate the distance in the two/one-dimensional space.
	Thus, the splitting receiver requires a higher computation complexity to achieve the performance gain.
	Regarding the circuit complexity, for each antenna branch, the splitting receiver requires two detection circuits, while the conventional CD/PD receiver only needs one detection circuit.
	On the other hand, the proposed simplified receiver has a lower complexity than the CD receiver and a higher complexity than the PD receiver.
	%
	Therefore, we should consider both the performance gain and the complexity (and the cost) when adopting a splitting receiver in practical systems.}

\subsection{Numerical Results}
In the last two subsections, we have shown and explained that the splitting receiver achieves a $3/2$ fold rate gain compared with the non-splitting channels in the high SNR regime.
This suggests that a notable performance improvement can be found within a moderate SNR range, which is verified as follows.
{\color{black}Also, we verify the tightness of the asymptotic analytical results presented in Sec.~III.B.}

\subsubsection{Single-antenna scenario}
We set the channel power gain $\vert\tilde{h}_1 \vert^2  =1$ for simplicity.

Fig.~\ref{fig:high_SNR_lower_bound} depicts the mutual information approximation given in \eqref{theory1_1} and also the simulated mutual information with different received signal power.
We see that the approximation and simulation results have the same general trend, and the percentage difference between the approximation and simulation results decreases as $\myP$ increase (e.g., from $\myP = 10$ to $100$).
{\color{black}Also we see that the optimal splitting ratios are (almost) the same for both the approximation and simulation results.
When $\snr$ is sufficiently large, e.g., $20$~dB, $\rho=0.33$ makes the mutual information at least $20\%$ larger than that of the conventional cases (i.e., $\rho = 0$ or $1$), and the joint processing gain is shown to be $G \approx 1.3$. 
When $\snr=30$~dB (i.e., $\myP=1000$), the approximation is tight, and the joint processing gain with $\rho = 0.33$ is close to $1.5$.
Thus, the tightness of the mutual information expressions in Lemma~1 and Proposition~1 (which is obtained by taking $\rho=1/3$ into Lemma~1) and also the asymptotic joint processing gain given by Proposition~3 are verified.}

\begin{figure*}[t]
	\small
	\renewcommand{\captionlabeldelim}{ }	
	\renewcommand{\captionfont}{\small} \renewcommand{\captionlabelfont}{\small}
	\minipage{0.47\textwidth}
	\centering
	\vspace*{-0.7cm}
	\includegraphics[width=\linewidth]{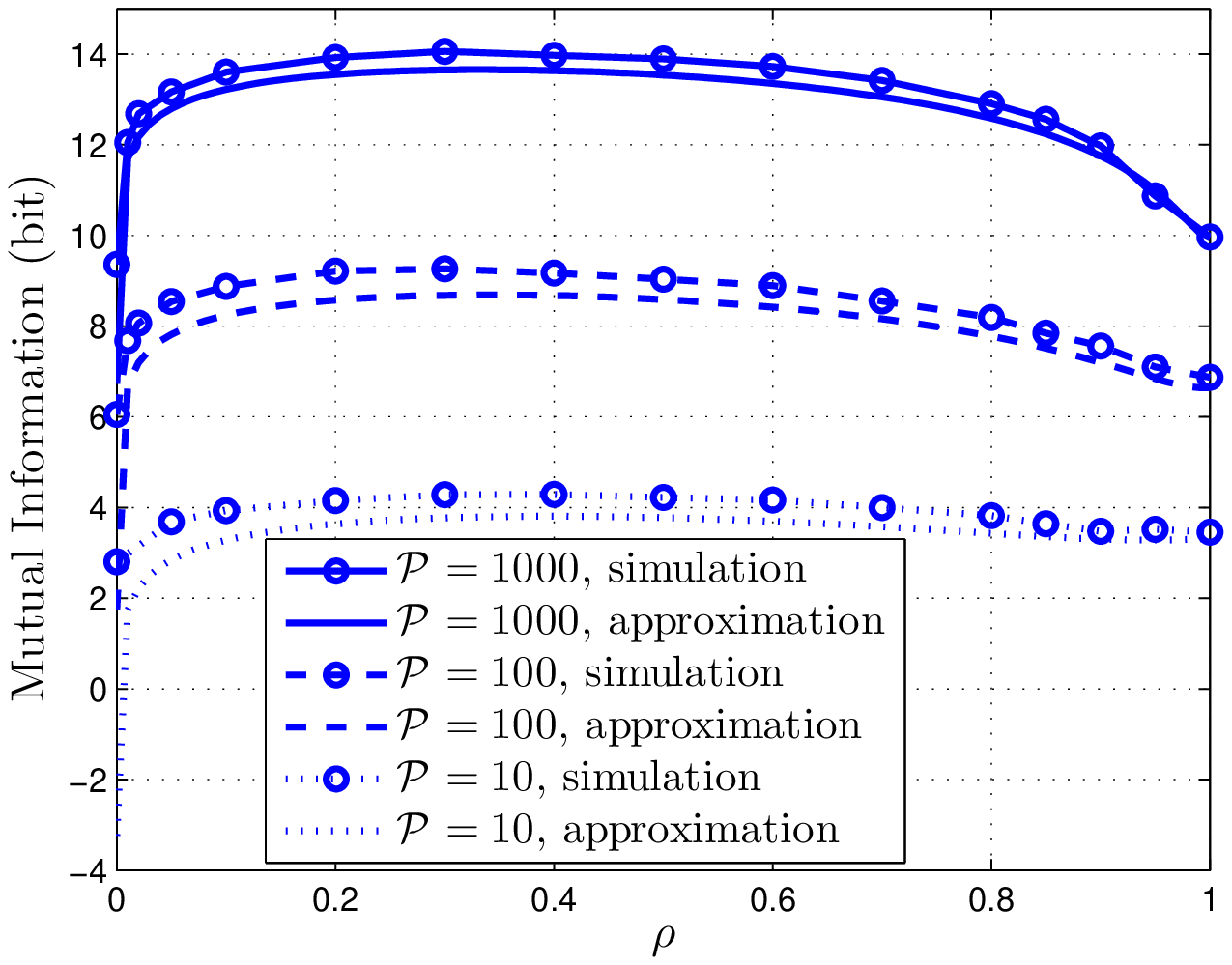}	
	\vspace*{-1.3cm}
	\caption{\small Mutual information versus $\rho$,\! $\sigmaone\! = \sigmatwo\! =1$.}	
	\label{fig:high_SNR_lower_bound}
	\endminipage
	\hspace{0.5cm}
	\minipage{0.47\textwidth}
	\centering
	\vspace*{-0.7cm}
	\includegraphics[width=\linewidth]{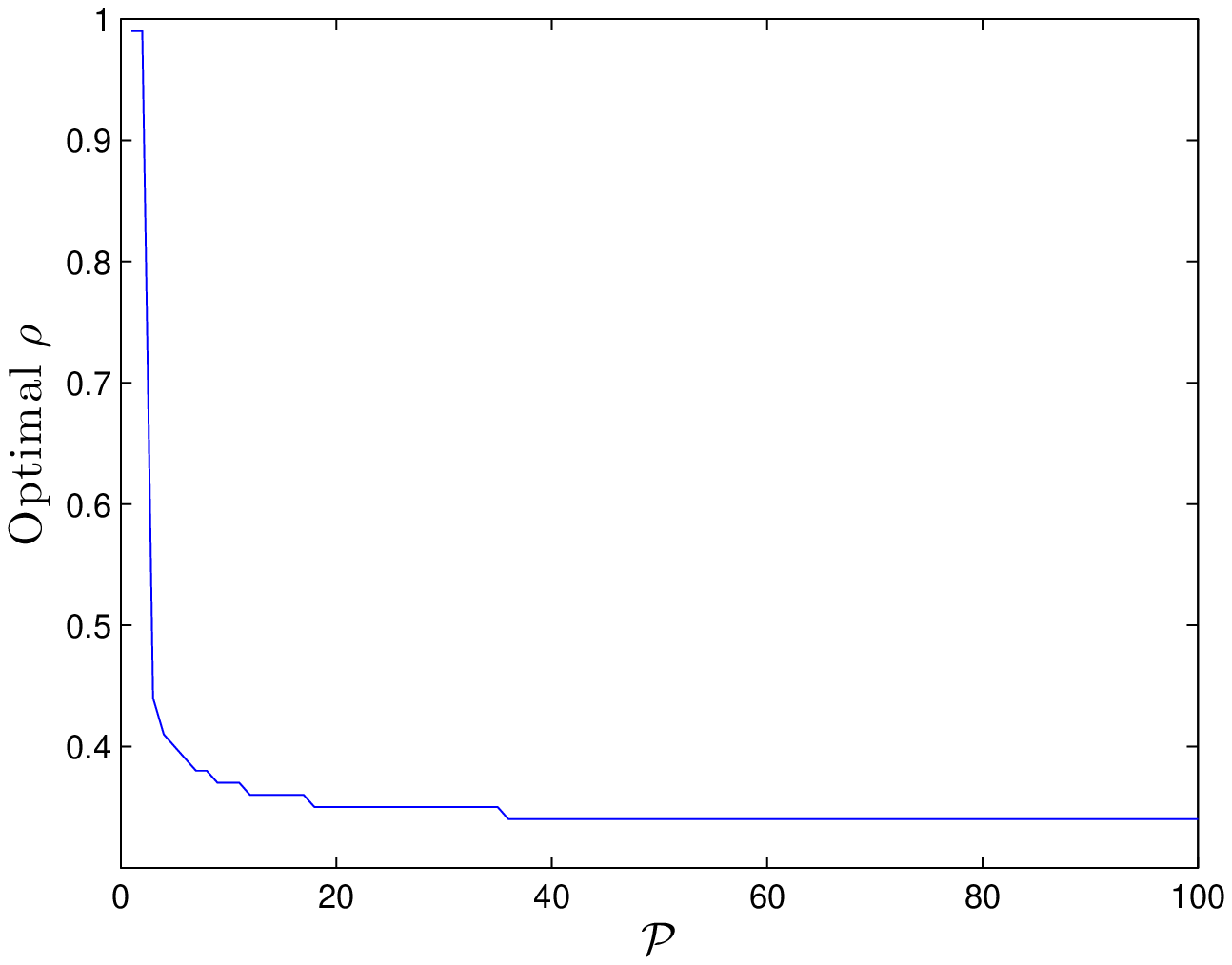}	
	\vspace*{-1.3cm}
	\caption{\small Optimal $\rho$ versus $\myP$, $\sigmaone = \sigmatwo = 1$.}	
	\label{fig:high_SNR_opti_rho}
	\endminipage
	\vspace*{-0.7cm}
\end{figure*}


Fig.~\ref{fig:high_SNR_opti_rho} depicts the optimal splitting ratio $\rho$ (obtained by simulation) versus the received signal power $\myP$. 
{\color{black}It is observed that the optimal splitting ratio approaches $1/3$ quickly as $\myP$ increases, i.e., $\rho=1/3$ when $\myP > 35$. \emph{Thus, 
	$\rho = 1/3$ is a near-optimal choice even at moderate SNRs,
	and the tightness of the asymptotic optimal splitting ratio in Proposition~1 is verified.}}


Fig.~\ref{fig:high_SNR_lower_bound_power} depicts the approximation of the splitting channel mutual information given in \eqref{theory1_1} with $\rho = 1/3$, and the optimal non-splitting channel mutual information, i.e., $\max\{\eqref{rho_1},\eqref{rho_0}\}$. 
It is observed that the splitting channel mutual information increases much faster w.r.t. $\myP$ as compared with the coherent and non-coherent AWGN channels.
When $\snr > 20$~dB (e.g., $\myP>100$, $\sigmaone =1$ and $\sigmatwo=0.1$, or $\myP>1000$, $\sigmaone =1$ and $\sigmatwo=10$), one can clearly see the mutual information improvement due to splitting. 

\begin{figure*}[t]
	\renewcommand{\captionfont}{\small} \renewcommand{\captionlabelfont}{\small}
	\minipage{0.47\textwidth}
	\includegraphics[width=\linewidth]{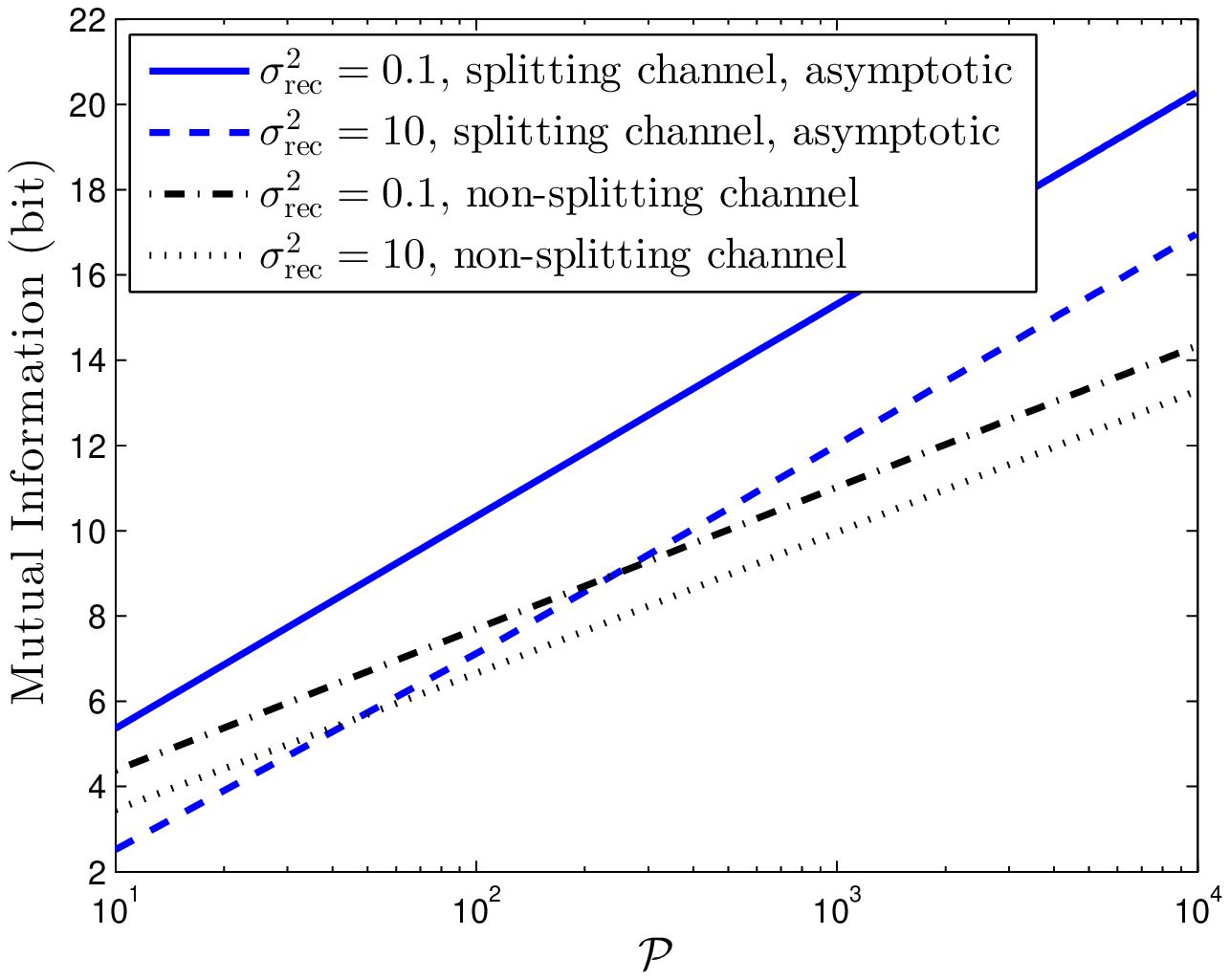}
	\vspace*{-1.3cm}
	\caption{Mutual information of the splitting channel and the non-splitting channel versus $\myP$, $\sigmaone = 1$.}\label{fig:high_SNR_lower_bound_power}
	\endminipage
	\hspace{0.5cm}
	\minipage{0.47\textwidth}
	\includegraphics[width=\linewidth]{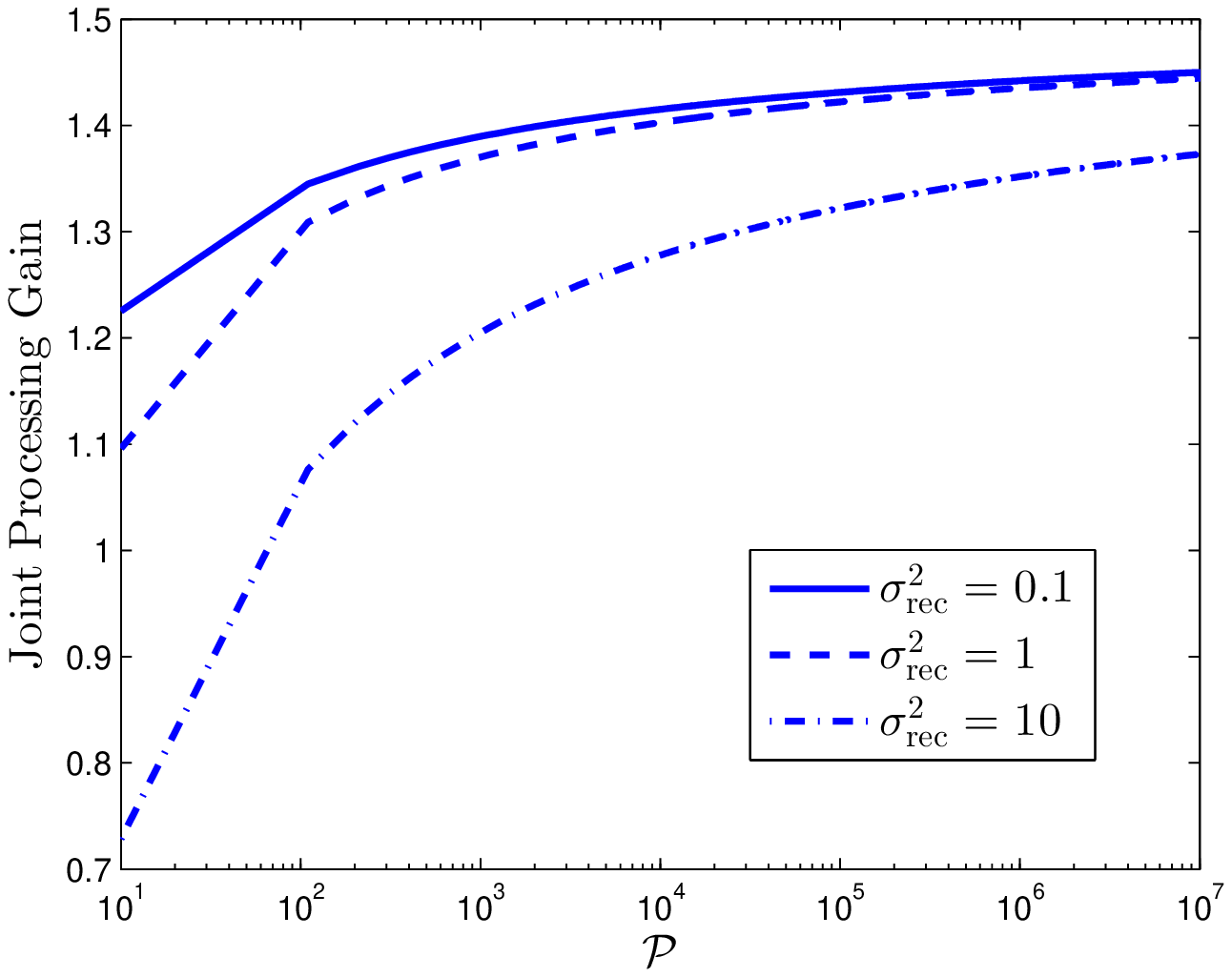}
	\vspace*{-1.3cm}
	\caption{Joint processing gain versus $\myP$, $\sigmaone = 1$.}\label{fig:splitting gain}
	\endminipage
\end{figure*}

Fig.~\ref{fig:splitting gain} depicts the joint processing gain obtained by taking \eqref{theory1_1} into  Definition~\ref{def:splitting_gain_MI}.
It is observed that the joint processing gain increases with $\myP$ and slowly approaches the constant~$3/2$. 
The gain at a practically high SNR, e.g., $30$~dB, is notable, e.g., $1.2 \sim 1.4$.
\subsubsection{Multi-antenna scenario}
Fig.~\ref{fig:K_4} depicts the average mutual information over $10^3$ channel realizations using \eqref{theory1_2}, where the channel power gain $\vert \tilde{h}_k \vert^2$ is assumed to follow an exponential distribution with the mean of $1$, and $\myP = 100$.
Three splitting strategies are considered:
(i) the numerically searched optimal splitting ratios by solving (P1) for every channel realization (i.e., optimal splitting), 
(ii) the simplified receiver with the strategy in Proposition~\ref{prop:large_antenna}, and 
(iii) $\rho_k = 1/3$ for all $k=1,2,...,K$.
It is observed that the splitting receiver with the optimal splitting strategy is better than the simplified receiver.
On the other hand, the splitting receiver can perform worse than the simplified receiver if some sub-optimal splitting strategy is used, e.g., $\rho = 1/3$, which is the optimal for a single-antenna receiver, but generally not necessarily optimal for a multi-antenna receiver.

Fig.~\ref{fig:large_K} depicts the optimal ratio of antennas allocated for coherent processing for a simplified receiver obtained by simulation using $10^4$ random channel realizations.
{\color{black}\emph{It shows that the optimal ratio is within the range of $(0.45,0.55)$ when $K>40$, and the optimal ratio converges to $1/2$ as $K$ increases further, which verifies Proposition~2.}}

\begin{figure*}[t]
	\renewcommand{\captionfont}{\small} \renewcommand{\captionlabelfont}{\small}
	\minipage{0.47\textwidth}
		\vspace*{-0.7cm}
	\includegraphics[width=\linewidth]{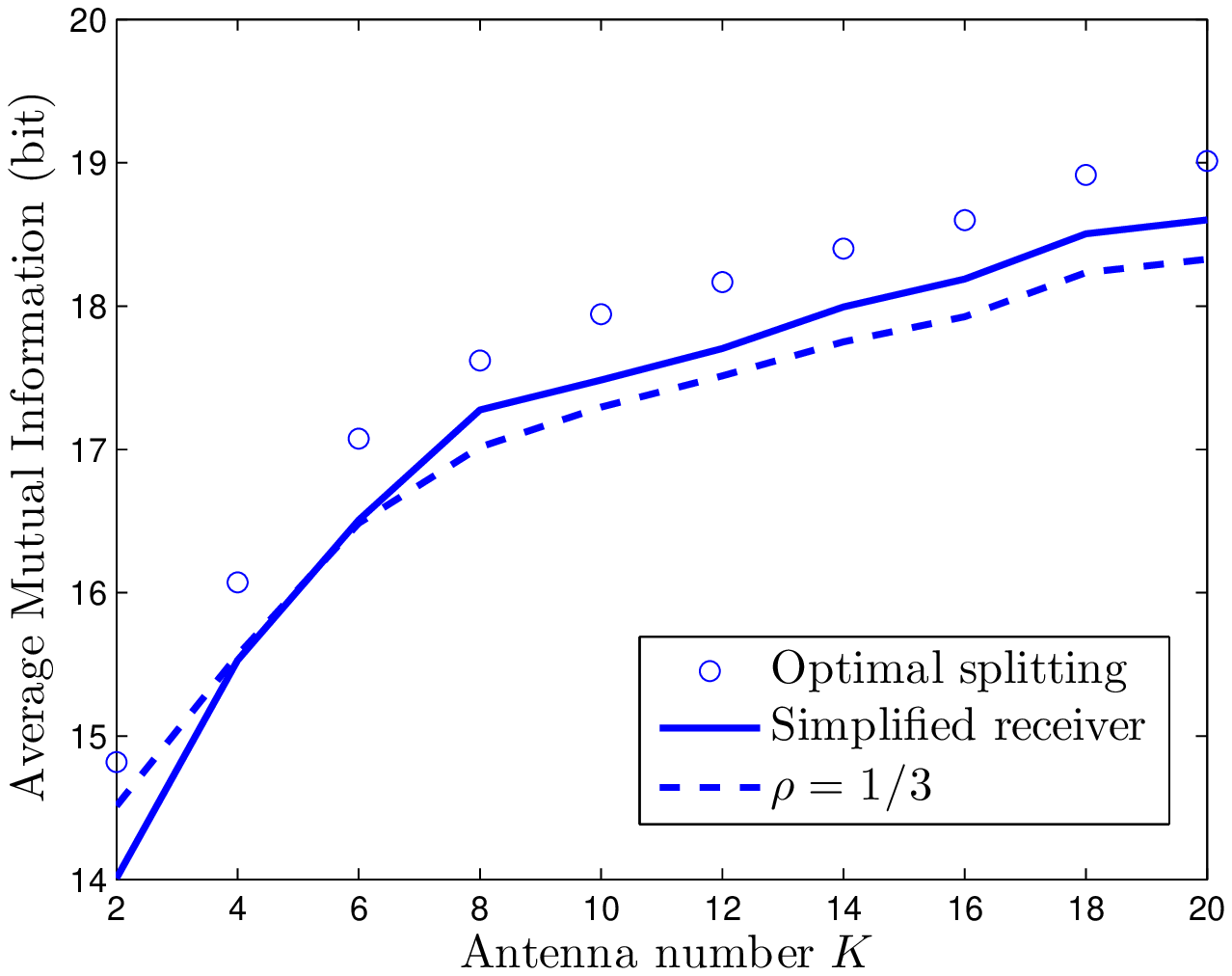}
	\vspace*{-1.3cm}
	\caption{Mutual information of the splitting channel with different splitting strategies, $\sigmaone = \sigmatwo = 1$.}\label{fig:K_4}
	\endminipage
	\hspace{0.5cm}
	\minipage{0.47\textwidth}
		\vspace*{-0.7cm}
	\includegraphics[width=\linewidth]{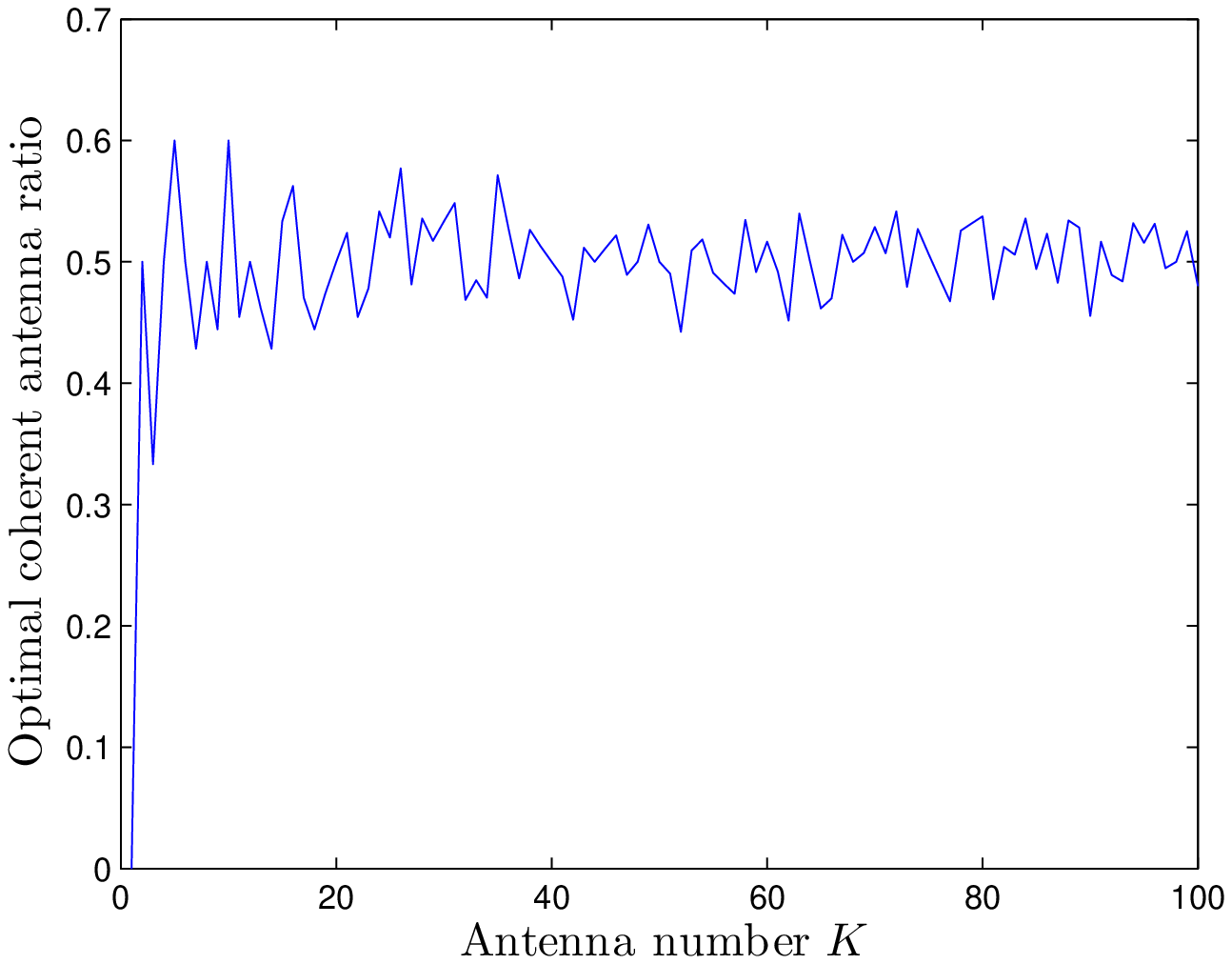}
	\vspace*{-1.3cm}
	\caption{Average optimal ratio of antennas allocated for coherent processing versus $K$.}\label{fig:large_K}
	\endminipage
	\vspace*{-0.7cm}
\end{figure*}

\section{Splitting Receiver: Practical Modulation}
In this section, we consider commonly used modulation schemes and assume that each signal of the constellation is transmitted with the same probability.
Note that, in this section, 
$x$ and $y$ denote the in-phase and quadrature signals in the CD circuit, respectively, and $z$ denotes the signal in the PD circuit, which are different from the notation in Sec.~III.
This change of notation is adopted for ease of presentation of results.

\subsection{Transmitted Signal Constellation}
We consider the transmitted signal constellation of a general $M$-ary modulation scheme is $\myomegagen$, which is a two-dimensional constellation placed on the I-Q plane.
The $i$th symbol is denoted by the tuple $({x}_i,{y}_i)$ on the I-Q plane, $i = 1,2,...,M$.
Specifically: 
\begin{enumerate} [(i)]
\item For the $M$-PAM scheme~\cite{BOOKTse}, which is a one-dimensional modulation scheme on the I-axis, we have ${x}_i= 2 i -1$ for $i=1,2,...,M/2$, and ${x}_i= - x_{i-M/2}$ for $i=M/2+1,...M$, and $y_i = 0$ for all $i$.
\item For the $M$-QAM scheme, which is a two-dimensional modulation scheme on the I-Q plane, we have ${x}_i= 2 \left(i \mod{\frac{\sqrt{M}}{2}} \right) - 1$, 
${y}_i= 2 \left\lfloor \frac{i-1}{\sqrt{M}/2}\right\rfloor  -1$,
$i=1,2,...,M/4$, which are the first quadrant symbols on the I-Q plane. Due to the symmetry property of $M$-QAM, the other quadrant symbol expressions are omitted for brevity.
\item For the $M$-IM scheme, which is a one-dimensional modulation scheme on the positive I-axis {\color{black}where} the information is carried by the signal power but not phase, we have $x_i = \sqrt{2(i-1)}$ and $y_i = 0$, $i=1,2,...,M$.
\end{enumerate}

\subsection{Noiseless Received Signal Constellation}
Based on the received signal expression after MRC in~\eqref{MRC_signal}, the average signal power of the coherently and non-coherently processed signals are $\Theta_1 \myP$ and $\sqrt{\Theta_2} \myP$, respectively.
Thus, with such average power constraints, we define the noiseless received signal constellation $\myomegagennew$,
and the $i$th symbol in $\myomegagennew$ is denoted by the tuple $(\breve{x}_i, \breve{y}_i, \breve{z}_i)$, and  
$\breve{x}_i = k_1 \sqrt{\Theta_1 \myP} {x}_i$, $\breve{y}_i = k_1 \sqrt{\Theta_1 \myP} {y}_i$, $\breve{z}_i =k_2 \sqrt{\Theta_2} \myP \left(x^2_i + y^2_i\right)$. Here $k_1$ and $k_2 \triangleq k^2_1$ are the power normalization parameters determined only by the geometric property of a certain modulation scheme. 
Specifically, we have 
\begin{equation} \label{geo_parameters}
k_1 = \left\lbrace 
\begin{aligned}
& \sqrt{\frac{3}{M^2-1}}, & M\text{-PAM},\\
& \sqrt{\frac{3}{2 (M-1)} }, & M\text{-QAM},\\
& \sqrt{\frac{1}{M-1}}, & M\text{-IM}.
\end{aligned}
\right.
\end{equation}
For the single-antenna scenario, 
Figs.~\ref{fig:PAM_rho_shape} and~\ref{fig:PPM_rho_shape}
show that the splitting ratio ${\rho} \in (0,1)$ bends the received signal constellation from the I-axis to a paraboloid in the I-P plane, for $4$-PAM and $4$-IM, respectively.
Fig.~\ref{fig:QAM_shape} shows that the splitting ratio ${\rho} \in (0,1)$ bends the received signal constellation from the I-Q plane to a paraboloid in the I-Q-P space.
\begin{figure}[t]
	\small
	\renewcommand{\captionlabeldelim}{ }
	\renewcommand{\captionfont}{\small} \renewcommand{\captionlabelfont}{\small}		
	\centering     
	\vspace{-0.7cm}
	\hspace{-0.5cm}
	\subfigure[$4$-PAM on the I-P plane.]{\label{fig:PAM_rho_shape}\includegraphics[scale=0.46]{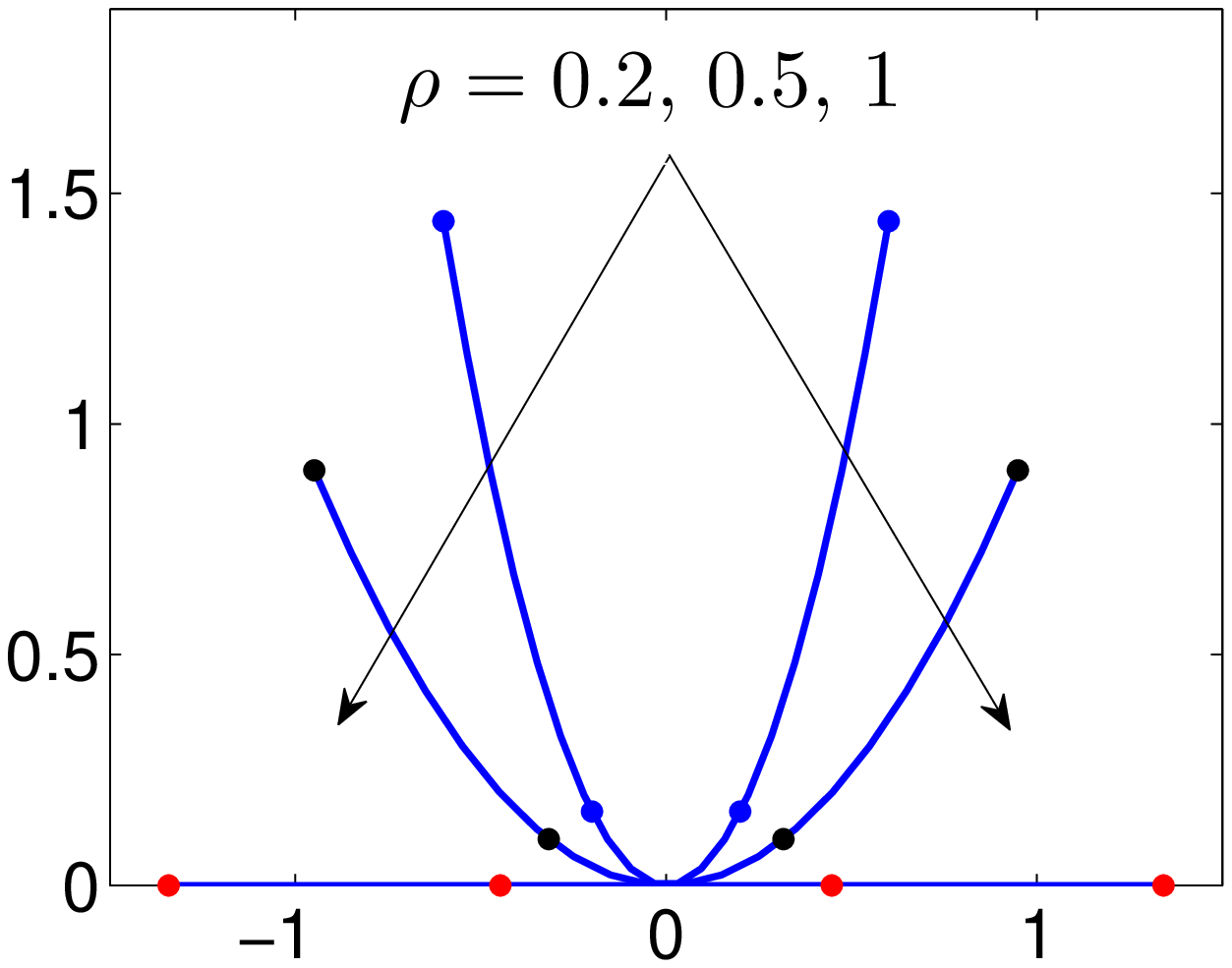}}
	\hspace{-0.6cm}
	\subfigure[$16$-QAM on the I-Q-P space, $\rho=0.2,\ 0.5,\ 1$.]{\label{fig:QAM_shape}\includegraphics[scale=0.72]{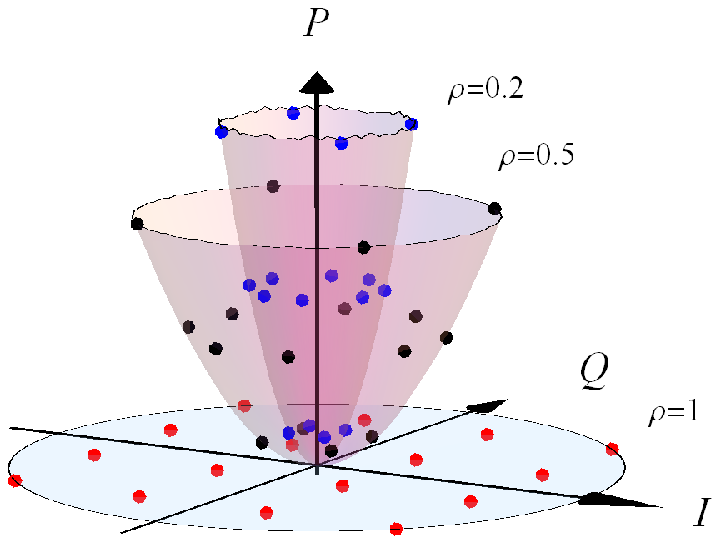}}
	\hspace{-0.4cm}
	\subfigure[$4$-IM on the I-P plane.]{\label{fig:PPM_rho_shape}\includegraphics[scale=0.5]{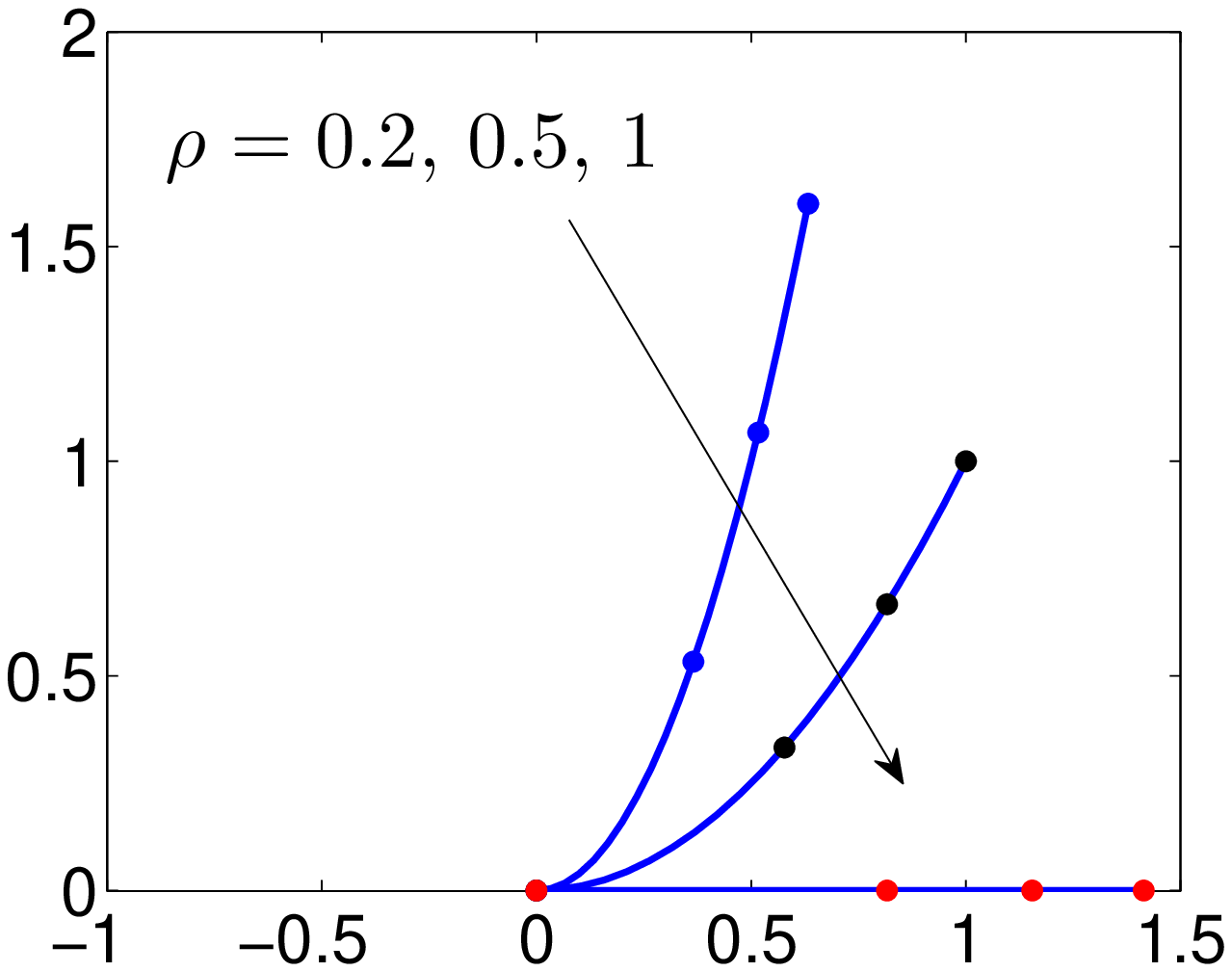}}	
	\caption{\small Noiseless received signal constellation with different $\vec{\rho}$, $K=1$, $\vert\tilde{h}_1\vert = 1$, $\myP = 10$.}
	\vspace{-0.5cm}
\end{figure}

\subsection{Decision Region}
Since all the transmitted symbols are of equal probability, the optimal signal detection method is the maximum likelihood (ML) method~\cite{BOOKTse}. The decision region for the $i$th symbol is defined~as
\begin{equation} \label{first_V}
\mathcal{V}_i \triangleq  \left\lbrace {\bm v} \vert f ({\bm v} \vert i) \geq f ({\bm v} \vert j) , \forall j \neq i, {\bm v} \in \mathbb{R}^3 \right \rbrace,
\end{equation}
where ${\bm v} \triangleq (x,y,z)$ is the three-dimensional post-processing (noise added) signal in the splitting receiver, and $f(\cdot \vert \cdot)$ is the conditional pdf.

From Sec. III, since both the CD and PD circuits introduce additive Gaussian noise,
the received signal is surrounded by the noise sphere (ellipsoid) in the I-Q-P space, and $f({\bm v}\vert i)$ is thus given by 
\begin{equation}
f({\bm v} \vert i) = \frac{1}{\sigmaone \pi \sqrt{2 \sigmatwo \pi}} \exp\left(
- \frac{(x-x_i)^2}{\sigmaone} 
- \frac{(y-y_i)^2}{\sigmaone}
- \frac{(z-z_i)^2}{2\sigmatwo}
\right).
\end{equation}
Therefore, \eqref{first_V} is rewritten as
\begin{equation} \label{second_V}
\mathcal{V}_i \triangleq  \left\lbrace (x,y,z):\ d_i(x,y,z) \leq d_j(x,y,z), \forall j\neq i\right\rbrace,
\end{equation}
where 
\begin{equation} \label{first_distance}
d_j(x,y,z) \triangleq \frac{\left(x-x_j\right)^2}{\sigmaone/2} + \frac{\left(y-y_j\right)^2}{\sigmaone/2} 
+  \frac{\left(z-z_j\right)^2}{\sigmatwo}.
\end{equation}
From \eqref{second_V} and \eqref{first_distance}, after simplification, the decision region of the $i$th symbol, $\mathcal{V}_i$ is given by
\begin{equation} \label{QAM_region}
\mathcal{V}_i= \left\lbrace (x,y,z):\  
\frac{\breve{x}_j-\breve{x}_i}{\sigmaone}\ x 
+\frac{\breve{y}_j-\breve{y}_i}{\sigmaone}\ y
+\frac{\breve{z}_j-\breve{z}_i}{2\sigmatwo}\ z
\leq \frac{\breve{x}^2_j+\breve{y}^2_j-\breve{x}^2_i-\breve{y}^2_i}{2\sigmaone} + \frac{\breve{z}^2_j-\breve{z}^2_i}{4 \sigmatwo}, \forall j \neq i
\right\rbrace,
\end{equation}
where $\breve{x}_i$, $\breve{y}_i$ and $\breve{z}_i$ are defined in Sec. IV.B above \eqref{geo_parameters}.
It is easy to see that $\mathcal{V}_i$ is bounded by planes. The plane implied in \eqref{QAM_region}, which divides the decision region between the $i$th and $j$th receive symbols, is given by
\begin{equation}
\mathcal{A}_{i-j} \triangleq \left\lbrace 
(x,y,z):\  
\frac{\breve{x}_j-\breve{x}_i}{\sigmaone}\ x 
+\frac{\breve{y}_j-\breve{y}_i}{\sigmaone}\ y
+\frac{\breve{z}_j-\breve{z}_i}{2\sigmatwo}\ z
= \frac{\breve{x}^2_j+\breve{y}^2_j-\breve{x}^2_i-\breve{y}^2_i}{2\sigmaone} + \frac{\breve{z}^2_j-\breve{z}^2_i}{4 \sigmatwo}
\right\rbrace.
\end{equation}

The decision regions for $8$-PAM, $36$-QAM (only for the symbols within the first quadrant of the I-Q-P space) and $4$-IM are illustrated in Figs.~\ref{fig:PAM_region},~\ref{fig:QAM_region}, and~\ref{fig:PPM_region}, respectively.

\subsection{Joint Processing Gain in SER}
To quantify the reduction in SER by the splitting receiver, we define the joint processing gain in terms of SER as:
\begin{definition}[Joint processing gain in SER] \label{def:splitting_gain}
Given a certain modulation scheme, the joint processing gain of the splitting receiver is
	\begin{equation}
	G \triangleq \frac{\min_{\vec{\rho} = \vec{0}, \vec{1}} P_e }{\inf\{P_e: \vec{\rho} \in \left[0,1 \right]^K\} },
	\end{equation}
	where $\inf\{\cdot\}$ denotes for the infimum, and $P_e$ is the SER for a given $\vec{\rho}$.
\end{definition}
The joint processing gain represents the maximum SER reduction provided by the splitting receiver, compared with the best of the conventional receivers.

\section{Splitting Receiver: SER Analysis}
In this section, we derive the SER at a splitting receiver for practical modulation schemes with the transmitted signal constellation $\myomegagen$ in the I-Q plane and the received signal constellation $\myomegagennew$ in the I-Q-P space.

The SER can be written as
\begin{equation}
P_{e} = \frac{1}{M}\sum\limits_{i=1}^{M} \left(1 - P_i\right), 
\end{equation}
where $P_i$ is the symbol success probability of the $i$th symbol, which is given by
\begin{equation} \label{QAM_SER}
P_i = \iiint_{\mathcal{V}_i} \exp\left(
-\frac{\left(x-\breve{x}_j\right)^2}{\sigmaone/2} - \frac{\left(y-\breve{y}_j\right)^2}{\sigmaone/2} 
-  \frac{\left(z-\breve{z}_j\right)^2}{\sigmatwo}  
\right) \mathrm{d}x\mathrm{d}y\mathrm{d}z.
\end{equation}

\begin{figure}[t]
	\small
	\renewcommand{\captionlabeldelim}{ }
	\renewcommand{\captionfont}{\small} \renewcommand{\captionlabelfont}{\small}		
	\centering     
	\vspace{-1cm}
	\subfigure[$8$-PAM decision regions in the I-P plane.]{\label{fig:PAM_region}\includegraphics[scale=0.35]{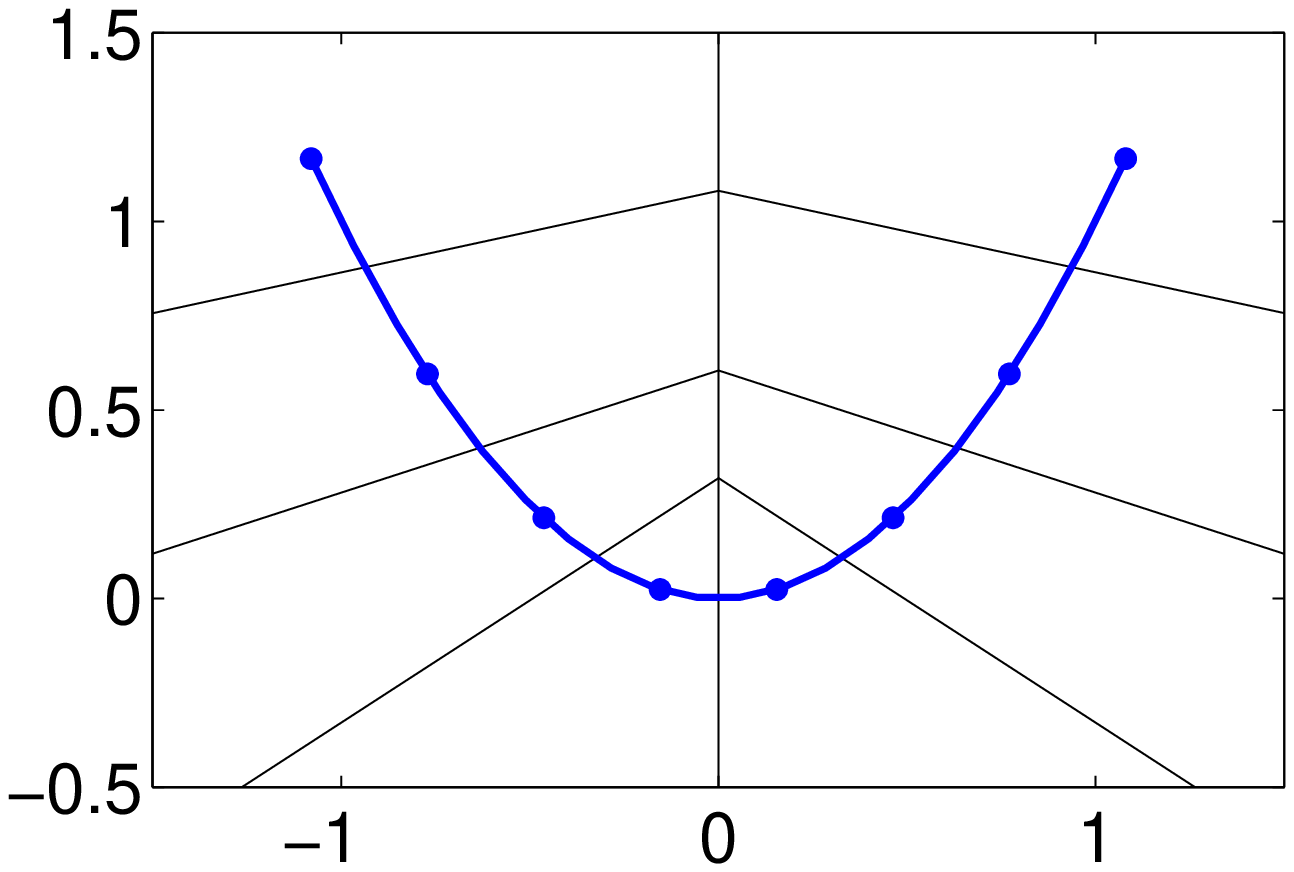}}
	\hspace{0.3cm}
	\subfigure[Illustration of the decision regions for $36$-QAM in the first quadrant, i.e., $x > 0$ and $y>0$.]{\label{fig:QAM_region}\includegraphics[scale=0.37]{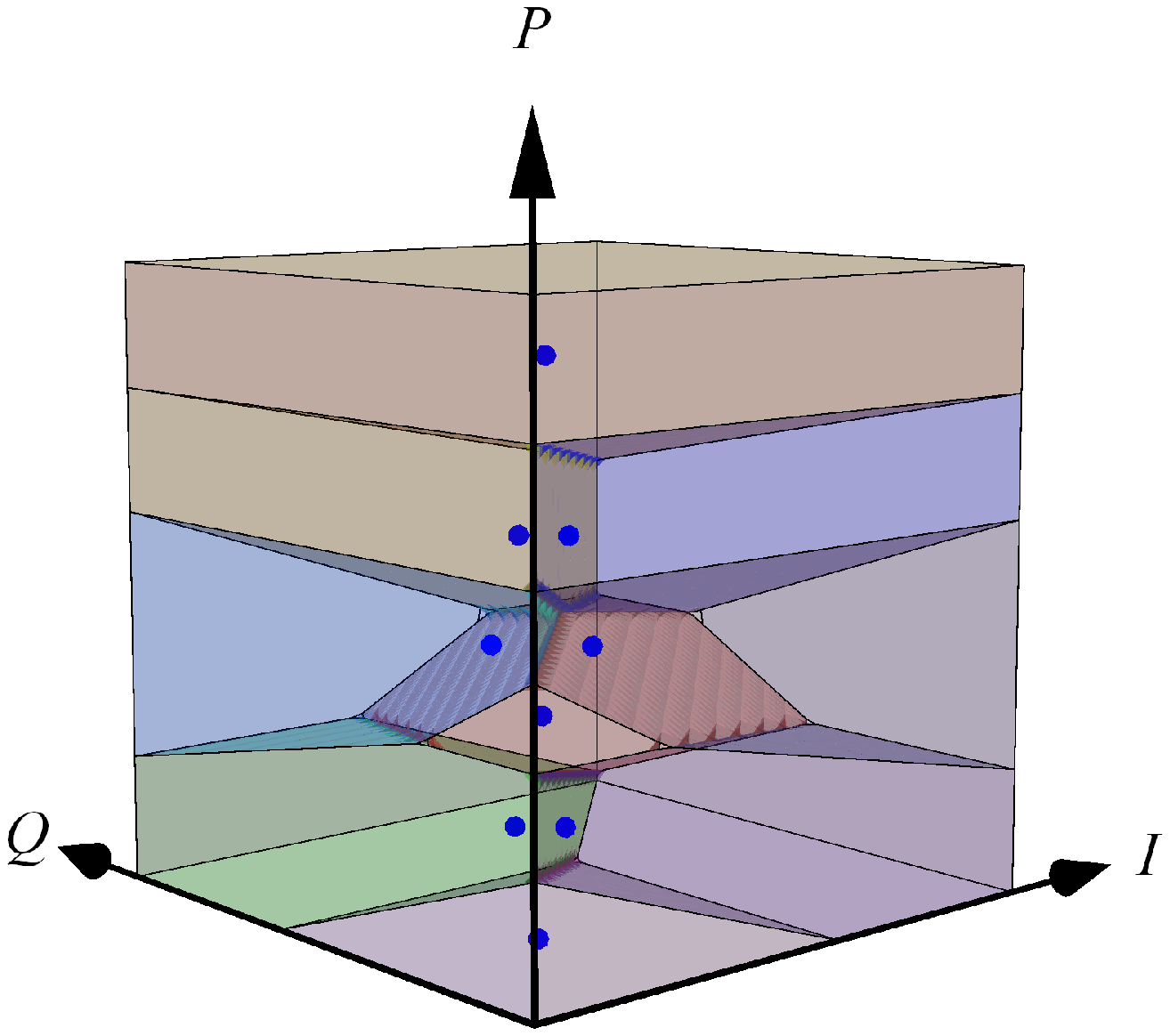}}
	\hspace{0.3cm}
	\subfigure[$4$-IM decision regions in the I-P plane.]{\label{fig:PPM_region}\includegraphics[scale=0.34]{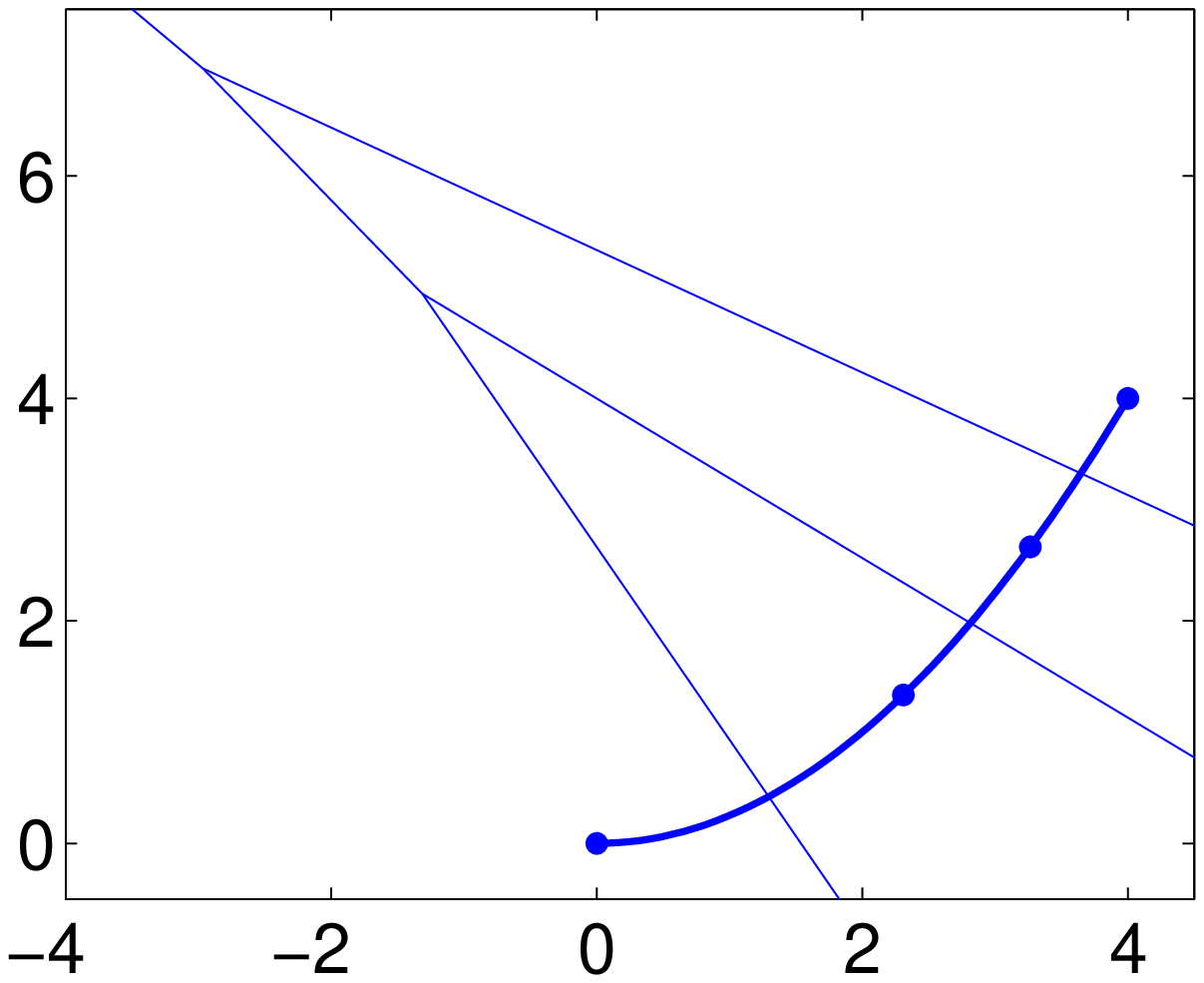}}	
	\caption{\small Decision regions for $8$-PAM, $36$-QAM and $4$-IM, $\myP = 10$, $\sigmaone =2$, $\sigmatwo =1$.}
	\label{fig:shapes}
	\vspace{-0.8cm}
\end{figure}

Based on Sec.~IV.B, 
when $\snr \rightarrow \infty$ and $\vec{\rho} \in \left[0,1 \right]^K\backslash\{\vec{0},\vec{1}\}$, 
the {\color{black}received} symbols $(\breve{x}_i,\breve{y}_i,\breve{z}_i)$ and $(\breve{x}_j,\breve{y}_j,\breve{z}_j)$ belonging to different power tiers, i.e., $\breve{z}_i \neq \breve{z}_j$, are easily distinguished because they are separated by a distance proportional to $\myP$ in the power domain.
In contrast, the symbols belonging to the same power tier are only separated with a distance proportional to $\sqrt{\myP}$ on the I-Q plane.
Thus, the intra-tier detection error probability dominates the overall SER in the high SNR regime.

Therefore, there are basically two cases for the SER analysis in the high SNR regime:
\begin{enumerate}
	\item For $\myomegagen$ having symbols that belong to the same tier as illustrated in Fig.~\ref{fig:general_constellation}(a), such as $M$-PAM, $M$-QAM and $M$-PSK (phase-shift keying), the intra-tier detection error probability is dominant. 
	Moreover, the detection error caused by the pair symbols with the minimum distance on the I-Q plane is dominant (see Fig.~\ref{fig:general_constellation}(a)).
	\item For $\myomegagen$ in which every symbol belongs to a different tier as illustrated in Fig.~\ref{fig:general_constellation}(b), such as a $M$-IM, the inter-tier detection error probability is dominant. Moreover, the detection error caused by the pair of symbols with the minimum distance in the P-axis (power domain) is dominant. 
\end{enumerate}

Consider a transmitted signal constellation $\myomegagen$ with $W$ pairs of dominant symbols as mentioned above and the minimum distance being $d_{\text{min}}$.
The approximated SER is calculated as~\cite{BOOKTse}
\begin{equation} \label{general_SER}
P_e \approx \frac{1}{M} \sum_{i=1}^{W} 2 Q\left( \frac{d_{\text{min}}}{2\sigma}\right),
\end{equation}
where $\sigma = \sqrt{\sigmaone/2}$ or $\sigmatwosqrt$ for cases 1) and 2), respectively.

{It is straightforward to see that}:
For $M$-PAM, we have $W=1$, i.e., the pair of symbols with the lowest power having the minimum distance given by $d_{\text{min}} = 2 \breve{x}_1$.
For $M$-QAM, we have $W=2 \sqrt{M}$, as illustrated in Fig.~\ref{fig:general_constellation}(c), and $d_{\text{min}}= 2 \breve{x}_1 $.
For $M$-IM, we have $W=M-1$, as illustrated in Fig.~\ref{fig:general_constellation}(b) and $d_{\text{min}} = \breve{x}_2 -\breve{x}_1$.
Then, based on \eqref{general_SER}, we can obtain the following results.

%

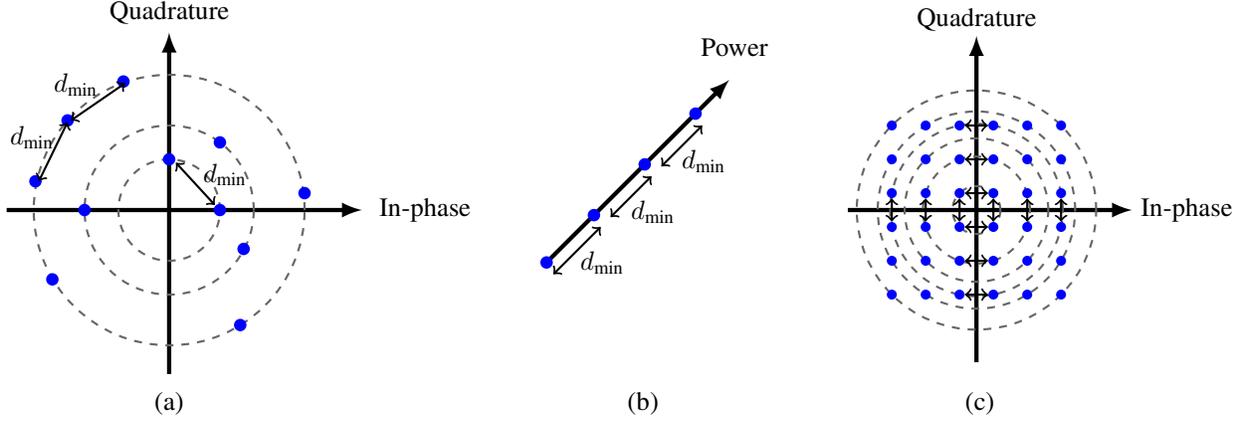
\begin{figure}[t]	
	\small
	\renewcommand{\captionlabeldelim}{ }
	\renewcommand{\captionfont}{\small} \renewcommand{\captionlabelfont}{\small}
	\centering
	\vspace{-0.7cm}
	\usetikzlibrary{arrows}
	\begin{tikzpicture}[scale = 0.45]
\draw [-latex, ultra thick](-4.8,0) -- (5.7,0);
\draw [-latex, ultra thick](0,-4.85) -- (0,5.25);

\draw [thick,dashed,black!60] (0,0) ellipse (4 and 4);
\draw [thick,dashed,black!60] (0,0) ellipse (2.5 and 2.5);
\draw [thick,dashed,black!60] (0,0) ellipse (1.5 and 1.5);

\node [circle,blue,fill,scale=0.5] at (1.5,2) {};
\node [circle,blue,fill,scale=0.5] at (2.2,-1.15) {};
\node [circle,blue,fill,scale=0.5] at (-2.5,0) {};
\node [circle,blue,fill,scale=0.5] at (1.5,0) {};
\node [circle,blue,fill,scale=0.5] at (0,1.5) {};
\node [circle,blue,fill,scale=0.5] at (-1.35,3.8) {}; 
\node [circle,blue,fill,scale=0.5] at (-3.45,-2.05) {};
\node [circle,blue,fill,scale=0.5] at (4,0.5) {};
\node [circle,blue,fill,scale=0.5] at (-3.95,0.85) {}; 
\node [circle,blue,fill,scale=0.5] at (-3,2.65) {}; 
\node [circle,blue,fill,scale=0.5] at (2.1,-3.4) {};
\node at (0,5.85) {Quadrature};
\node at (7.55,0) {In-phase};

\draw [-latex,<->,thick](-2.9,2.65) -- (-1.3,3.75);
\draw [-latex,<->,thick](-3.85,0.85) -- (-3,2.6);
\draw [-latex,<->,thick](0.2,1.4) -- (1.35,0.15);
\node at (-4.05,2.15) {$d_{\text{min}}$};
\node at (-2.75,3.7) {$d_{\text{min}}$};
\node at (1.65,1.05) {$d_{\text{min}}$};


\draw [-latex, ultra thick](11.15,-1.55) -- (16.55,3.85);
\node [circle,blue,fill,scale=0.5] at (11.15,-1.55) {};
\node [circle,blue,fill,scale=0.5] at (12.55,-0.15) {};
\node [circle,blue,fill,scale=0.5] at (14.05,1.35) {};
\node [circle,blue,fill,scale=0.5] at (15.55,2.85) {};

\node at (16.7,4.8) {Power};
\draw [-latex,<->,thick](11.35,-1.9) -- (12.75,-0.5) node (v1) {};
\draw [-latex,<->,thick](v1) -- (14.25,1) node (v2) {};
\draw [-latex,<->,thick](v2) -- (15.75,2.5) node (v3) {};

\node at (12.8,-1.55) {$d_{\text{min}}$};
\node at (14.3,-0.05) {$d_{\text{min}}$};
\node at (15.8,1.45) {$d_{\text{min}}$};


\node at (0,-5.7) {(a)};
\node at (14,-5.7) {(b)};

\draw [-latex, ultra thick](20.05,0) -- (28.45,0);
\draw [-latex, ultra thick](23.85,-4.5) -- (23.85,5.15);
\draw [thick,dashed,black!60] (23.85,0) ellipse (0.707 and 0.707);
\draw [thick,dashed,black!60] (23.85,0) ellipse (1.581 and 1.581);
\draw [thick,dashed,black!60] (23.85,0) ellipse (2.549 and 2.549);

\draw [thick,dashed,black!60] (23.85,0) ellipse ( 2.1213 and 2.1213);
\draw [thick,dashed,black!60] (23.85,0) ellipse (  2.9155 and 2.9155);
\draw [thick,dashed,black!60] (23.85,0) ellipse ( 3.5355 and 3.5355);

\node [circle,blue,fill,scale=0.4] (v11) at (21.35,0.5) {};
\node [circle,blue,fill,scale=0.4] (v9) at (22.35,0.5) {};
\node [circle,blue,fill,scale=0.4] (v1) at (23.35,0.5) {};
\node [circle,blue,fill,scale=0.4] (v7) at (26.35,0.5) {};
\node [circle,blue,fill,scale=0.4] (v5) at (25.35,0.5) {};
\node [circle,blue,fill,scale=0.4] (v2) at (24.35,0.5) {};

\node [circle,blue,fill,scale=0.4] at (21.35,1.5) {};
\node [circle,blue,fill,scale=0.4] at (22.35,1.5) {};
\node [circle,blue,fill,scale=0.4] (v13) at (23.35,1.5) {};
\node [circle,blue,fill,scale=0.4] at (26.35,1.5) {};
\node [circle,blue,fill,scale=0.4] at (25.35,1.5) {};
\node [circle,blue,fill,scale=0.4] (v21) at (24.35,1.5) {};

\node [circle,blue,fill,scale=0.4] at (21.35,2.5) {};
\node [circle,blue,fill,scale=0.4] at (22.35,2.5) {};
\node [circle,blue,fill,scale=0.4] (v14) at (23.35,2.5) {};
\node [circle,blue,fill,scale=0.4] at (26.35,2.5) {};
\node [circle,blue,fill,scale=0.4] at (25.35,2.5) {};
\node [circle,blue,fill,scale=0.4] (v22) at (24.35,2.5) {};

\node [circle,blue,fill,scale=0.4] at (21.35,-2.5) {};
\node [circle,blue,fill,scale=0.4] at (22.35,-2.5) {};
\node [circle,blue,fill,scale=0.4] (v17) at (23.35,-2.5) {};
\node [circle,blue,fill,scale=0.4] at (26.35,-2.5) {};
\node [circle,blue,fill,scale=0.4] at (25.35,-2.5) {};
\node [circle,blue,fill,scale=0.4] (v18) at (24.35,-2.5) {};

\node [circle,blue,fill,scale=0.4] at (21.35,-1.5) {};
\node [circle,blue,fill,scale=0.4] at (22.35,-1.5) {};
\node [circle,blue,fill,scale=0.4] (v15) at (23.35,-1.5) {};
\node [circle,blue,fill,scale=0.4] at (26.35,-1.5) {};
\node [circle,blue,fill,scale=0.4] at (25.35,-1.5) {};
\node [circle,blue,fill,scale=0.4] (v16) at (24.35,-1.5) {};

\node [circle,blue,fill,scale=0.4] (v12) at (21.35,-0.5) {};
\node [circle,blue,fill,scale=0.4] (v10) at (22.35,-0.5) {};
\node [circle,blue,fill,scale=0.4] (v3) at (23.35,-0.5) {};
\node [circle,blue,fill,scale=0.4] (v8) at (26.35,-0.5) {};
\node [circle,blue,fill,scale=0.4] (v6) at (25.35,-0.5) {};
\node [circle,blue,fill,scale=0.4] (v4) at (24.35,-0.5) {};

\node at (23.85,5.65) {Quadrature};
\node at (30.05,0) {In-phase};

\draw [<->,thick](v1) -- (v2);
\draw [<->,thick](v3) -- (v1);
\draw [<->,thick](v3) -- (v4);
\draw [<->,thick](v2) -- (v4);
\draw [<->,thick](v5) -- (v6);
\draw [<->,thick](v7) -- (v8);
\draw [<->,thick](v9) -- (v10);
\draw [<->,thick](v11) -- (v12);
\draw [<->,thick](v13) -- (v21);
\draw [<->,thick](v14) -- (v22);
\draw [<->,thick](v15) -- (v16);
\draw [<->,thick](v17) -- (v18);

\node at (23.95,-5.7) {(c)};
	\end{tikzpicture}
	\vspace{-0.5cm}		
	\caption{(a) and (b) Two transmitted signal constellation maps with $3$ pairs of symbols that are dominant on detection error probability, plotted on the I-Q plane and the P-axis, respectively.
	(c) The transmitted signal constellation for $36$-QAM, where $12$ pairs of symbol that are dominant on detection error probability are illustrated.}
	\label{fig:general_constellation}	
	\vspace{-0.5cm}
\end{figure}

\subsection{$M$-PAM}
\begin{proposition}  \label{lemma_opti_rho}
	For the $M$-PAM scheme and $\vec{\rho} \in \left[0,1 \right]^K\backslash\{\vec{0},\vec{1}\}$, the SER in the high SNR regime is given by
	\begin{equation}
	{P}_e 
	\approx
	\frac{2}{M} Q\left(
	\frac{\sqrt{2}\breve{x}_1}{\sigmaonesqrt}
	\right).
	\end{equation}
	\end{proposition}
	Based on Proposition~\ref{lemma_opti_rho}, 
	we can see that when $\vec{\rho} \neq \vec{1}$ and $\vec{\rho} \rightarrow \vec{1}$, $\breve{x}_1 \approx \sqrt{3 H_2 \myP /(M^2-1)}$, and 
	$P_e \approx \frac{2}{M}  Q\left(
	\sqrt{\frac{6 H_2 \myP}{\sigmaone (M^2-1)}}
	\right)
	$,
	which is smaller than the SER of the $\vec{\rho} = \vec{1}$ case, i.e., $\frac{2(M-1)}{M} \times\\ Q\left(
	\sqrt{\frac{6 H_2 \myP}{\sigmaone (M^2-1)}}\right)$, and is also smaller than the SER of the $\vec{\rho} = \vec{0}$ case in which the SER can be as large as $0.5$.
	Thus, we have the following proposition:
	\begin{proposition}\label{asym_splitting_gain}
	For $M$-PAM, the asymptotic joint processing gain in the high SNR regime~is 
	\begin{equation}
	G_{\mathrm{PAM}} = \frac{\min\left\lbrace 0.5, \frac{2 (M-1)}{M}Q\left(\sqrt{\frac{6 H_2 \myP}{\sigmaone (M^2-1)}}\right) \right\rbrace}{\frac{2}{M}Q\left(\sqrt{\frac{6 H_2 \myP}{\sigmaone (M^2-1)}}\right)}
	= M-1.
	\end{equation}
	\end{proposition}
	Note that although the joint processing gain depends on $\myP$, $\sigmaone$ and $\sigmatwo$, the asymptotic joint processing gain is independent of the specific noise variance at the CD and PD circuits in the high SNR regime.

\subsection{$M$-QAM}
\begin{proposition}\label{QAM_high_SNR}
	For $M$-QAM and $\vec{\rho} \in \left[0,1 \right]^K\backslash\{\vec{0},\vec{1}\}$, the SER in the high SNR regime is given by
	\begin{equation} \label{QAM_approx}
	P_e \approx \frac{4}{\sqrt{M}} Q\left( \frac{\sqrt{2}\breve{x}_1}{\sigmaonesqrt} \right).
	\end{equation}
\end{proposition}

Letting $\vec{\rho} \rightarrow \vec{1}$, i.e., $\breve{x}_1 \rightarrow \sqrt{\frac{3 H_2 \myP}{2(M-1)}}$, the approximated SER in Proposition~\ref{QAM_high_SNR} is minimized as $P_e \approx \frac{4}{\sqrt{M}} Q \left(\sqrt{\frac{3 H_2 \myP}{(M-1) \sigmaone}} \right)$, which is smaller than the SER obtained by setting $\vec{\rho} = \vec{0}$ or $\vec{1}$.
Thus, we have the following result:
\begin{proposition}\label{Splitting_gain_QAM}
	For $M$-QAM, the asymptotic joint processing gain in the high SNR regime is 
	\begin{equation}
	\begin{aligned}
	G_{\mathrm{QAM}} 
	&= \lim\limits_{\snr \rightarrow \infty}\frac
	{4 \left(1 -\frac{1}{\sqrt{M}}\right) Q\left( \sqrt{\frac{3 H_2 \myP}{(M-1) \sigmaone}} \right)
		-4 \left(1 -\frac{1}{\sqrt{M}}\right)^2 Q\left(\sqrt{ \frac{3 H_2 \myP}{(M-1) \sigmaone} }\right)^2
	}
	{\frac{4}{\sqrt{M}}   Q\left(\sqrt{\frac{3 H_2 \myP}{(M-1) \sigmaone}}\right) }\\
	&= \lim\limits_{\snr \rightarrow \infty} \frac
	{4 \left(1 -\frac{1}{\sqrt{M}}\right) Q\left( \sqrt{\frac{3 H_2 \myP}{(M-1) \sigmaone}} \right)}
	{\frac{4}{\sqrt{M}}   Q\left(\sqrt{\frac{3 H_2 \myP}{(M-1) \sigmaone}}\right) }
	= \sqrt{M}-1.
	\end{aligned}
	\end{equation}
\end{proposition}
Therefore, in the high SNR regime, for $M$-PAM and $M$-QAM, there always exists a non-trivial $\vec{\rho}$ that $\vec{\rho}\in \left[0,1 \right]^K \backslash \{\vec{0},\vec{1}\}$ and achieves a lower SER than the conventional receivers, i.e., $\vec{\rho} = \vec{0} \text{ or }\vec{1}$, no matter what values $\sigmaone$ and $\sigmatwosqrt$ take.

\subsection{$M$-IM}
\begin{proposition} \label{PPM_asyn_SER}
	For $M$-IM, the SER in the high SNR regime is given by
	\begin{equation} 
	P_e \approx  \frac{2 (M-1)}{M} Q\left(\frac{\sqrt{\Theta_2}\myP}{(M-1) \sigmatwosqrt}\right).
	\end{equation}
\end{proposition}

From Proposition \ref{PPM_asyn_SER}, as $\vec{\rho} \rightarrow \vec{0}$, the minimum approximated SER is obtained as $P_e = \frac{2 (M-1)}{M} \times \\ Q\left(\frac{\sqrt{H_4} \myP}{(M-1) \sigmatwosqrt}\right)$, which equal to the SER when $\vec{\rho} = \vec{0}$. Thus, the splitting receiver cannot improve the SER performance compared with the conventional receivers, and we have the result:
\begin{proposition} \label{PPM}
	For $M$-IM, the asymptotic joint processing gain in the high SNR regime is equal to one.
\end{proposition}

\subsection{Numerical Results}
We present the numerical results using $M$-QAM for (i) the splitting receiver with single receiver antenna assuming $\vert \tilde{h}_1 \vert^2 = 1$, and (ii) the simplified receiver with multiple receiver antennas.
The SER results for $M$-QAM are plotted based on Monte Carlo simulation with $10^9$ points using the detection rule~\eqref{QAM_region}.
The results for $M$-PAM and $M$-IM are omitted due to space limitation.
\subsubsection{Splitting receiver with single receiver antenna}
Fig. \ref{fig:QAM_diff_M_P} plots the SER versus the splitting ratio $\rho$ for different $M$ and different $\myP$, where the approximation results are plotted using Proposition~\ref{QAM_high_SNR}.
It shows that the SER first decreases and then increases as $\rho$ increases.
We can see that the optimal $\rho$ that minimizes the SER, increases with $\myP$ and approaches $1$ but decreases with the increasing of the order of constellation $M$.
{\color{black}Also we can see that when $\snr$ is sufficiently large, e.g., $23$~dB (i.e., $\myP = 200$, and $\sigmaone = \sigmatwo = 1$), the approximation of the SER is very close with the accurate SER for the value of $\rho$ in the range $(0,\rho^\star)$, where $\rho^\star$ is the optimal splitting ratio. 
	Note that $\rho^\star$ approaches $1$ as $\snr$ increases. This means the mismatch around $\rho =1$ is minimized as $\snr$ increases.
	\emph{Therefore, the approximation in Proposition 6 is accurate for the values of $\rho \in (0,1)$ when $\snr$ is sufficiently large.}
	}


Fig. \ref{fig:QAM_splitting_gain} shows the joint processing gain versus $\myP$ using Definition~\ref{def:splitting_gain}.
We can see that the joint processing gain increase with $\myP$ and approaches $3$ and $5$ for $16$-QAM and $36$-QAM, respectively, when $\myP = 100$, $\sigmaone = 1$ and $\sigmatwo$ is sufficiently small, e.g., $10^{-3}$. These results approach the asymptotic joint processing gain in Proposition~\ref{Splitting_gain_QAM}.
Also we see that only half the joint processing gain is achieved when $\myP = 100$ and $\sigmatwo$ is large, e.g., $\sigmatwo=1$.
However, the increasing trend of the joint processing gain in Fig.~15 suggest that the asymptotic joint processing gain can be eventually achieved, when $\myP$ is much larger than $100$.
%
{\color{black}\emph{Therefore, the asymptotic joint processing gain in Proposition~7 may not be approached in a normal range of the received signal power and the noise variance, but half of the joint processing gain is achievable.}}

\begin{figure*}[t]
	\small
	\renewcommand{\captionlabeldelim}{ }	
	\renewcommand{\captionfont}{\small} \renewcommand{\captionlabelfont}{\small}
	\minipage{0.47\textwidth}
	\vspace*{-0.7cm}
	\includegraphics[width=\linewidth]{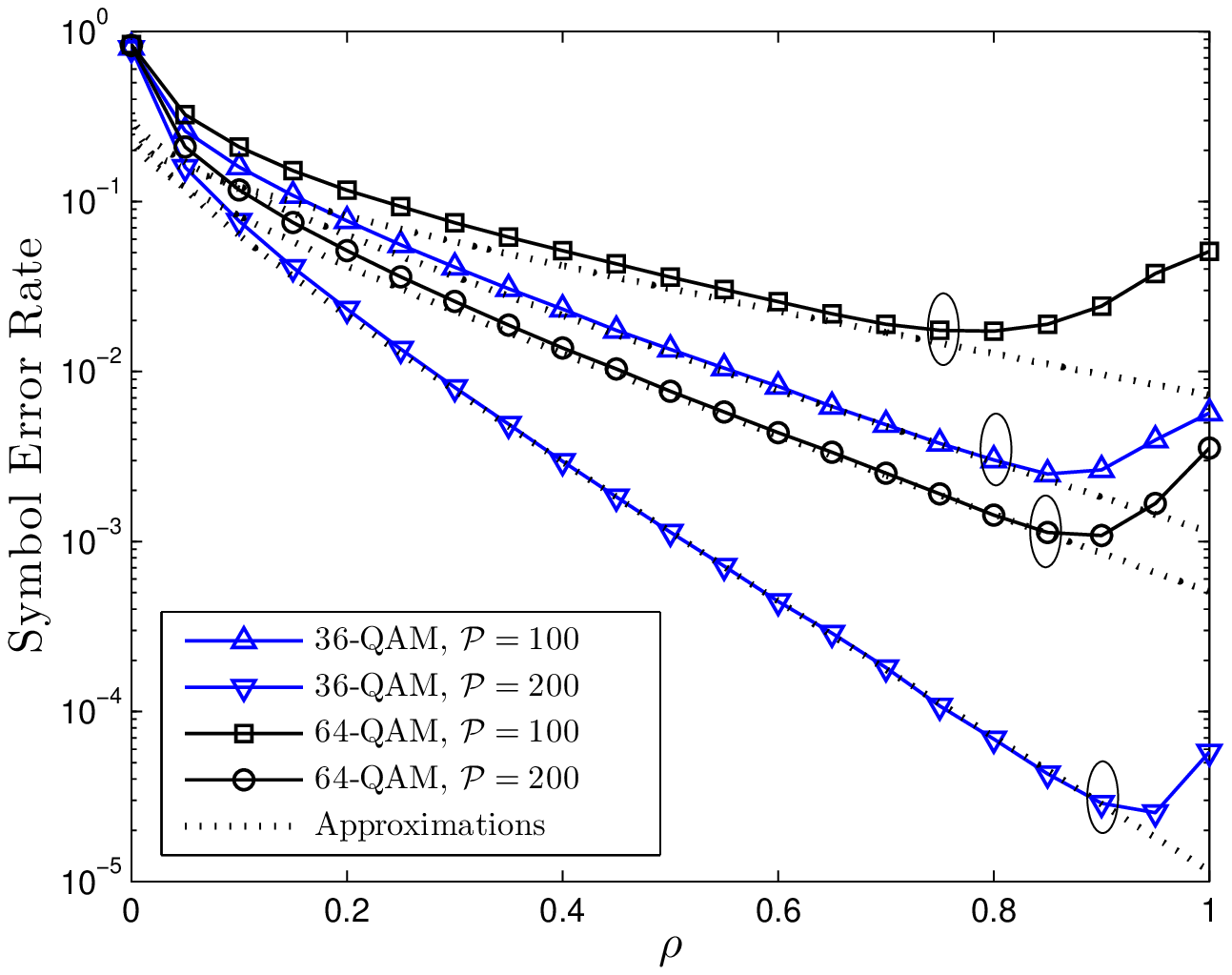}	
	\vspace*{-1.3cm}
	\caption{\small SER versus $\rho$, $\sigmaone = 1$ and $\sigmatwo = 1$.}	
	\label{fig:QAM_diff_M_P}
	\endminipage
	\hspace{0.32cm}
	\minipage{0.47\textwidth}
	\vspace*{-0.7cm}
	\includegraphics[width=\linewidth]{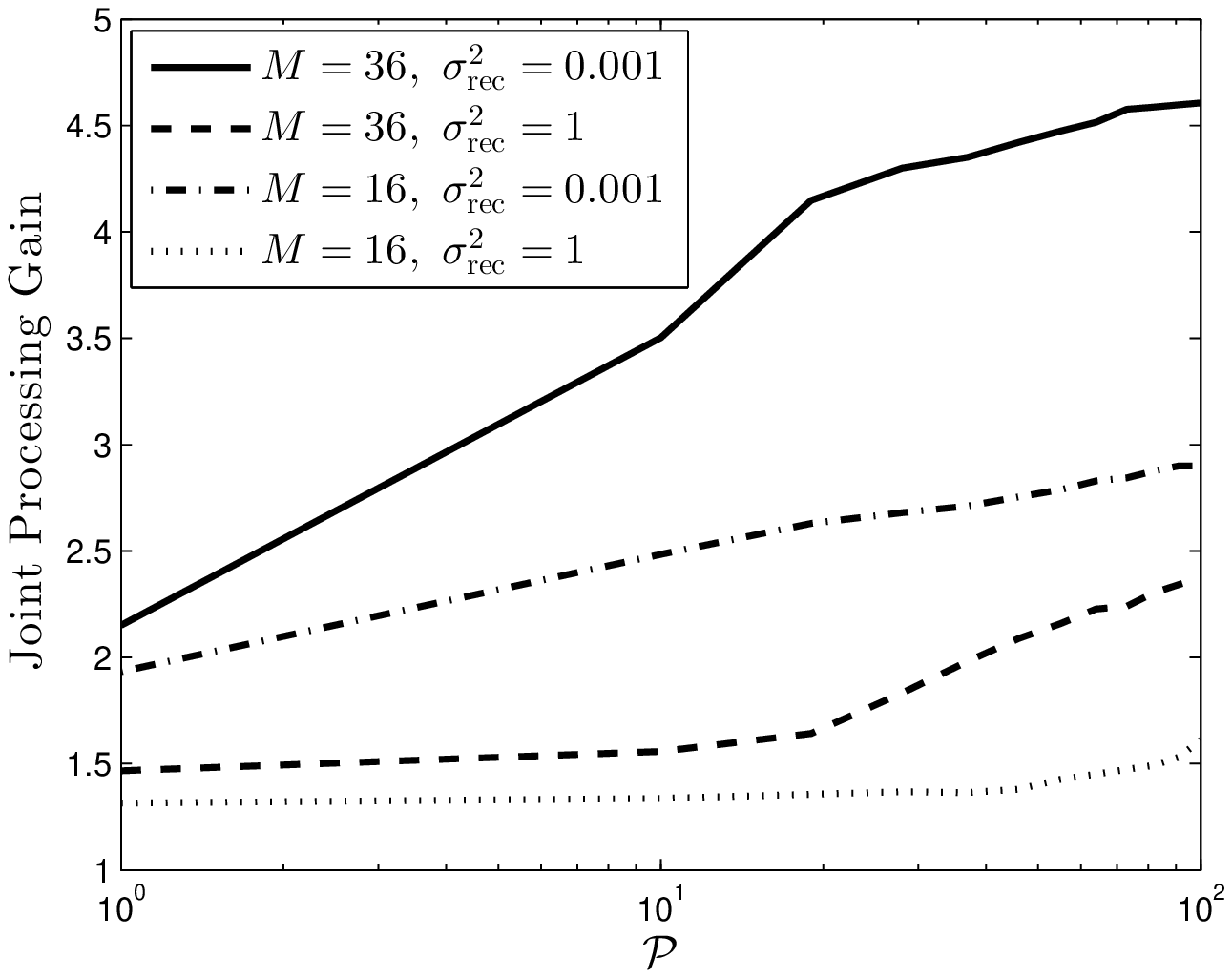}	
	\vspace*{-1.3cm}
	\caption{\small $G_{\text{QAM}}$ versus $\myP$, $\sigmaone = 1$.}	
	\label{fig:QAM_splitting_gain}
	\endminipage
	\vspace*{-0.5cm}
\end{figure*}


\subsubsection{Simplified receiver with multiple receiver antennas}
For the simplified receiver, we assume that the channel power gain at each antenna is independent and follows an exponential distribution with the mean of $1$.
Fig.~\ref{fig:modu_large_K} plots the optimal number of antennas allocated for coherent processing, i.e., $K^{\star}_1$, versus the total number of antennas for the $36$-QAM scheme obtained by using $10^3$ random channel realizations.
It shows that $K^{\star}_1$ increases with $K$, and approaches to $K-1$ in the high SNR regime,
e.g, $K^{\star}_1 \approx K-5,\ K-2$ and $K-1$ when $\myP = 2,\ 10$ and $200$, respectively.
This is because that the optimal ratio $\vec{\rho} \rightarrow \vec{1}$ but never reaches $\vec{1}$ based on Proposition~\ref{lemma_opti_rho} in the high SNR regime. 
In other words, for the simplified receiver (where $\rho_k \in\{0,1\}$), most of the antennas should be connected to the CD circuits and at least one antenna should be connected to a PD circuit to achieve the highest joint processing gain.

{\color{black}Note that in practice, the degradation of ADC noise is usually modeled by the signal-to-quantization-noise ratio (SQNR), approximately given by $6K$~dB, where $K$ is the number of quantization bits~\cite{Xunzhou13}. Here, by assuming $\myP = 2$ and the noise variance of the ADC equals to $0.1$ (i.e., less than $\sigmaone$ and $\sigmatwo$), the SQNR equals to $13$~dB, which implies
	$K \approx 2$~bits. Similarly, by assuming $\myP = 200$, the SQNR equals to $33$~dB, which implies $K \approx 5$~bits. 
	Therefore, the parameter settings are practical.}

\section{Conclusions}
In this paper, we have proposed a splitting receiver, which fundamentally changes the way in which the signal is processed. 
With the same received signal power, the analytical results show that the splitting receiver provides excellent performance gain in the sense of both the mutual information (Gaussian input) and the SER (practical modulation), compared with the conventional coherent and non-coherent receivers.
Future research may focus on the topics such as the MIMO system with a multi-antenna splitting receiver, and the design of constellations and coding schemes for the communication systems with splitting receiver. 
Moreover, some practical issues can also be taken into account, such as the effects of the antenna noise, the power splitter losses and the different receive sensitivity level at the CD and PD circuits of the splitting receiver.

\begin{figure}[t]
	\small
	\renewcommand{\captionlabeldelim}{ }
	\renewcommand{\captionfont}{\small} \renewcommand{\captionlabelfont}{\small}	
	\centering 
	\vspace*{-0.7cm}	
	\includegraphics[scale=0.6]{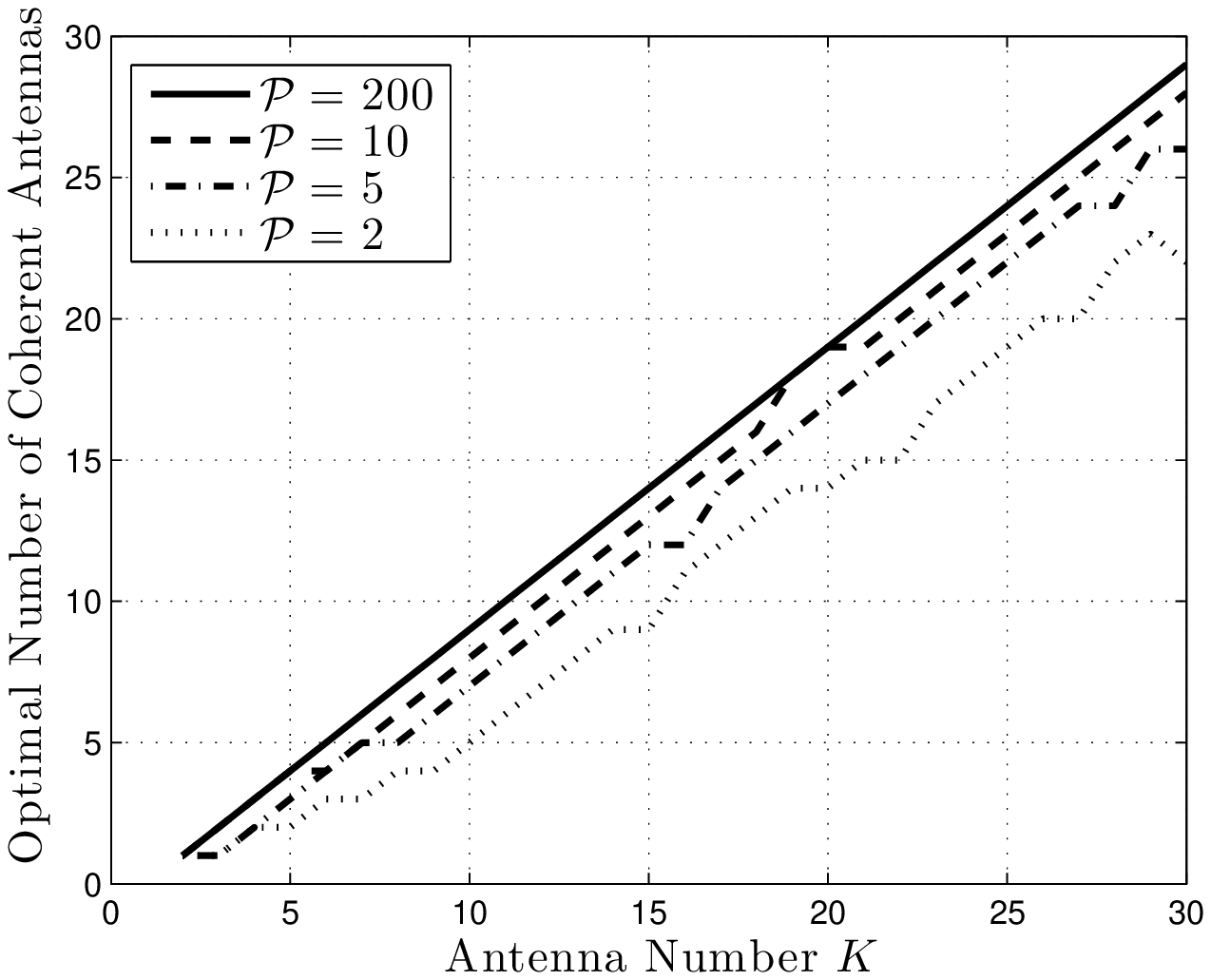}	
	\vspace*{-0.5cm}
	\caption{\small Optimal number of antennas allocated for coherent processing versus $K$ for $16$-QAM, $\sigmaone = \sigmatwo = 1$.}	
	\label{fig:modu_large_K}
	\vspace*{-0.5cm}
\end{figure}

\begin{appendices}
	\renewcommand{\theequation}{\thesection.\arabic{equation}}
	\numberwithin{equation}{section}

\section{Proof of Lemma~1}
We assume that $\vec{\rho} \in \left[0,1\right]{\color{black}^K} \backslash \{\vec{0},\vec{1}\}$. Thus, based on \eqref{my_theta}, $\Theta_1 >0$ and $\Theta_2 >0$.
\subsection{Proof of \eqref{theory1_1}}
Based on the property of mutual information invariance under scaling of random variables~\cite{Invariance},
a scaled received signal expression based on \eqref{receive_signal} is given by
\begin{equation} \label{rewritten_signal}
\begin{aligned}
\tilde{Y}_1 & = \sqrt{\Theta_1} \tilde{X} + \frac{\tilde{Z}}{\sqrt{\myP}},\\
Y_2 &= \sqrt{\Theta_2} k \sqrt{\myP} \vert \tilde{X} \vert^2 + k \frac{N}{\sqrt{\myP}},
\end{aligned}
\end{equation}
where $k \triangleq \frac{\sigmaonesqrt}{\sqrt{2}\sigmatwosqrt}$.
Thus, it is easy to verify that the real and imaginary parts of $\frac{\tilde{Z}}{\sqrt{\myP}}$ and $k \frac{{\color{black}N}}{\sqrt{\myP}}$ are independent with each other and follow the same distribution $\mathcal{N}(0,\frac{\sigmaone}{2 \myP})$.

We define two random variables as 
\begin{equation} \label{new_X1_X2}
\tilde{X}_1 \triangleq \sqrt{\Theta_1} \tilde{X} \text{, and } X_2 \triangleq \sqrt{\Theta_2}k \sqrt{\myP} \vert \tilde{X} \vert^2.
\end{equation}
Because of the Markov chain $\sqrt{\myP}{\color{black}\tilde{X}} \rightarrow (\tilde{X}_1,X_2) \rightarrow (\tilde{Y}_1, Y_2)$ and 
the smooth and uniquely invertible map from $\sqrt{\myP}\tilde{X}$ to $(\tilde{X}_1,X_2)$,
we have 
\begin{equation}
\mathcal{I}(\sqrt{\myP} X; \tilde{Y}_1, Y_2) = \mathcal{I}(\tilde{X}_1,X_2; \tilde{Y}_1, Y_2).
\end{equation}

Before the analysis of $ \mathcal{I}(\tilde{X}_1,X_2; \tilde{Y}_1, Y_2)$, we first define a new coordinate system named as paraboloid-normal (PN) coordinate system which is based on a paraboloid $\mathcal{U}$. The paraboloid $\mathcal{U}$ is defined by the equation
\begin{equation}
c_P =k \sqrt{\myP} \frac{\sqrt{\Theta_2}}{\Theta_1} (c^2_I+c^2_Q),
\end{equation}
where $c_I$, $c_Q$ and $c_P$ are the three axes of Cartesian coordinate system of the I-Q-P space.  
By changing coordinate system, the point $(c_1,c_2,c_3)$ is represented as $(\tilde{a},l)$ in the PN coordinate system, where $\tilde{a}$ is the nearest point on the paraboloid $\mathcal{U}$ to the point $(c_1,c_2,c_3)$, and $\vert l\vert$ is the distance.
In other words, the point $(c_1,c_2,c_3)$ is on the normal line at the point $\tilde{a}$ on the paraboloid. Specifically, the sign of $l$ is positive when the point $(c_1,c_2,c_3)$ is above the parabolic, otherwise, it is negative.

Based on the property of mutual information invariance under a change of coordinates~\cite{Invariance},
representing Cartesian coordinate based random variables $(\tilde{X}_1,X_2)$ and $(\tilde{Y}_1, Y_2)$ under the PN coordinate system as $(\tilde{A}_X, L_X)$ and $(\tilde{A}_X + \tilde{A}_{\tilde{Z},N},L_{\tilde{Z},N})$, respectively, gives
\begin{equation}
\mathcal{I}(\tilde{X}_1,X_2; \tilde{Y}_1, Y_2) = \mathcal{I}(\tilde{A}_X, L_X; \tilde{A}_X + \tilde{A}_{\tilde{Z},N},L_{\tilde{Z},N}),
\end{equation}
where the noise-related random variables $\tilde{A}_{\tilde{Z},N}$ and $L_{\tilde{Z},N}$, which are generated by $\tilde{Z}$ and $N$, are correlated with the random variable $\tilde{A}_X$.
%
%
Since $X_2= k\sqrt{\myP}\frac{\sqrt{\Theta_2}}{\Theta_1} \vert \tilde{X}_1 \vert^2$, $(\tilde{X}_1,X_2)$ lies on the paraboloid $\mathcal{U}$, i.e., $L_X$ is a constant which is equal to zero, $\tilde{A}_{X}$ can be represented by $(\tilde{X}_1,X_2)$ for brevity. 
Thus, we have
\begin{equation} \label{MI_coordinates}
\begin{aligned}
\mathcal{I}(\tilde{X}_1,X_2; \tilde{Y}_1, Y_2) 
&= \mathcal{I}(\tilde{A}_X; \tilde{A}_X + \tilde{A}_{\tilde{Z},N},L_{\tilde{Z},N})\\
&= h(\tilde{A}_X + \tilde{A}_{\tilde{Z},N},L_{\tilde{Z},N}) - h(\tilde{A}_X + \tilde{A}_{\tilde{Z},N},L_{\tilde{Z},N} \vert \tilde{A}_X)\\
&= h(\tilde{A}_X + \tilde{A}_{\tilde{Z},N}) + h(L_{\tilde{Z},N} \vert \tilde{A}_X + \tilde{A}_{\tilde{Z},N})\\
&-\left(
h(\tilde{A}_X + \tilde{A}_{\tilde{Z},N}\vert \tilde{A}_X) + 
h(L_{\tilde{Z},N} \vert \tilde{A}_X, \tilde{A}_X + \tilde{A}_{\tilde{Z},N})
\right).
\end{aligned}
\end{equation}
{\color{black}Since the expectations $\myexpect{\tilde{Z}}=(0,0)$ and $\myexpect{N}=0$, and the variances $\mathrm{Var}(\frac{\tilde{Z}}{\sqrt{\myP}}) \rightarrow 0$ and $\mathrm{Var}(k \frac{N}{\sqrt{\myP}}) \rightarrow 0$ as $\myP \rightarrow \infty$,
it is easy to see that the noise variable $\tilde{A}_{\tilde{Z},N}$ converges in probability towards $(0,0)$.
Thus, $\tilde{A}_X + \tilde{A}_{\tilde{Z},N}$ converges in probability towards $\tilde{A}_X$.

Furthermore, since convergence in probability implies convergence in distribution and the entropy function $h(\cdot)$ is continuous and defined based on the probability distribution of the input random variable~\cite{BookInfo}, 
we have $h(\tilde{A}_X + \tilde{A}_{\tilde{Z},N}) \rightarrow h(\tilde{A}_X)$ as $\myP \rightarrow \infty$.
Similarly, we have the convergence of the random variable, i.e., $\left(L_{\tilde{Z},N} , \tilde{A}_X + \tilde{A}_{\tilde{Z},N}\right) \rightarrow \left(L_{\tilde{Z},N} , \tilde{A}_X\right)$, hence the convergence of entropy, i.e.,
$h(L_{\tilde{Z},N} , \tilde{A}_X + \tilde{A}_{\tilde{Z},N}) \rightarrow h(L_{\tilde{Z},N} , \tilde{A}_X)$.

Therefore, the conditional entropy  
$
h(L_{\tilde{Z},N} \vert \tilde{A}_X + \tilde{A}_{\tilde{Z},N}) \triangleq h(L_{\tilde{Z},N} , \tilde{A}_X + \tilde{A}_{\tilde{Z},N}) - h(\tilde{A}_X + \tilde{A}_{\tilde{Z},N}) 
$
converges to $h(L_{\tilde{Z},N} \vert \tilde{A}_X) \triangleq h(L_{\tilde{Z},N} , \tilde{A}_X) - h(\tilde{A}_X)$, i.e., $h(L_{\tilde{Z},N} \vert \tilde{A}_X + \tilde{A}_{\tilde{Z},N})
\rightarrow 
h(L_{\tilde{Z},N} \vert \tilde{A}_X)$. 
Similarly, we have $h(L_{\tilde{Z},N} \vert \tilde{A}_X, \tilde{A}_X + \tilde{A}_{\tilde{Z},N}) \rightarrow h(L_{\tilde{Z},N} \vert \tilde{A}_X, \tilde{A}_X)$.
Together with the fact that $h(L_{\tilde{Z},N} \vert \tilde{A}_X)= h(L_{\tilde{Z},N} \vert \tilde{A}_X, \tilde{A}_X)$~\cite{BookInfo}, we have
$h(L_{\tilde{Z},N} \vert \tilde{A}_X + \tilde{A}_{\tilde{Z},N}) - h(L_{\tilde{Z},N} \vert \tilde{A}_X, \tilde{A}_X + \tilde{A}_{\tilde{Z},N})  \rightarrow 0$ as $\myP \rightarrow \infty$. }
Thus, the mutual information in \eqref{MI_coordinates} can {\color{black}asymptotically} be rewritten as
\begin{equation} \label{MI_asymptotic_appen}
\mathcal{I}(\tilde{X}_1,X_2; \tilde{Y}_1, Y_2) = h(\tilde{A}_X) - h(\tilde{A}_X + \tilde{A}_{\tilde{Z},N}\vert \tilde{A}_X).
\end{equation}

Then we calculate $h(\tilde{A}_X)$ and $h(\tilde{A}_X + \tilde{A}_{\tilde{Z},N}\vert \tilde{A}_X)$ as follows.

\subsubsection{$h(\tilde{A}_X)$}
Due to the fact that the probability contained in a differential area should not alter under a change of variables, we have 
\begin{equation} \label{first_derri_eq}
\vert f_{\tilde{X}}(\tilde{x}) \mathrm{d} S \vert = \vert f_{\tilde{A}_{X}}(\tilde{a}) \mathrm{d} \Sigma \vert ,
\end{equation}
where 
$\mathrm{d} S = \mathrm{d} u \mathrm{d} v$, 
$u = \mathrm{Real}\{\tilde{x}\}$, $v = \mathrm{Imag}\{\tilde{x}\}$, 
$\tilde{a}$ is the PN coordinate system representation of the point $(\tilde{x}_1,x_2)$,
$\mathrm{d} \Sigma$ is the differential area on the paraboloid $\mathcal{U}$, 
$f_{\tilde{A}_{X}}(\tilde{a})$ and $f_{\tilde{X}}(\tilde{x})$ are the pdfs of $\tilde{A}_{X}$ and $\tilde{X}$, respectively,
and
\begin{equation} \label{normal_pdf}
f_{\tilde{X}}(\tilde{x}) = \frac{1}{\pi} \exp\left( - \vert \tilde{x} \vert^2\right).
\end{equation}

Assuming that ${\bf r} = (\mathrm{Real}\{\tilde{x}_1\},\mathrm{Imag}\{\tilde{x}_1\}, x_2)$, which is a point on the paraboloid $\mathcal{U}$, and thus, based on \eqref{new_X1_X2}, we have 
\begin{equation} \label{my_Sigma}
\begin{aligned}
\frac{\partial {\bf r}}{\partial u} &= (\sqrt{\Theta_1 },0,2 k\sqrt{\myP}\sqrt{\Theta_2} u),\ 
\frac{\partial {\bf r}}{\partial v} = (0,\sqrt{\Theta_1 },2 k\sqrt{\myP} \sqrt{\Theta_2} v),\\
\mathrm{d} \Sigma &= \left\vert \frac{\partial {\bf r}}{\partial u} \times  \frac{\partial {\bf r}}{\partial v} \right\vert \mathrm{d} u \mathrm{d}v
=\Theta_1  \sqrt{4 \frac{k^2 \myP \Theta_2}{\Theta^2_1} \vert \tilde{x}_1\vert^2 + 1} \ \mathrm{d} u \mathrm{d}v,
\end{aligned}
\end{equation}
where $\times$ is the cross product operator.
Taking \eqref{my_Sigma}, \eqref{normal_pdf} and $\tilde{x} = \frac{\tilde{x}_1}{\sqrt{\Theta_1}}$ into \eqref{first_derri_eq}, after simplification, we have
\begin{equation} \label{f_Y1_Y2}
f_{\tilde{A}_{X}}(\tilde{a})  = f_{\tilde{A}_{X}}(\tilde{x}_1,x_2)  = \frac{1}{\pi \Theta_1  \sqrt{4 \frac{k^2 \myP \Theta_2}{\Theta^2_1} \vert \tilde{x}_1\vert^2 + 1}} \exp\left( - \frac{\vert \tilde{x}_1 \vert^2}{\Theta_1 }\right).
\end{equation}

The differential entropy of $\tilde{A}_{X}$ is derived as
\begin{equation} \label{joint_entropy}
\begin{aligned}
h(\tilde{A}_{X})
&= \iint - f_{\tilde{A}_{X}}(\tilde{a}) \log_2\left(f_{\tilde{A}_{X}}(\tilde{a})\right) \mathrm{d}\Sigma.
\end{aligned}
\end{equation}

Taking \eqref{f_Y1_Y2} and \eqref{my_Sigma} into \eqref{joint_entropy}, we have
\begin{equation} \label{highSNR_entropy_2}
\begin{aligned}
h(\tilde{A}_{X}) &= 
\int_{-\infty}^{\infty} \int_{-\infty}^{\infty} -
\frac{\exp\left( - \frac{\vert \tilde{x}_1 \vert^2}{\Theta_1 }\right)}{\pi} 
\log_2
\left(
\frac{\exp\left( - \frac{\vert \tilde{x}_1 \vert^2}{\Theta_1 }\right)}{\pi \Theta_1  \sqrt{4 \frac{k^2 \myP\Theta_2}{\Theta^2_1} \vert \tilde{x}_1\vert^2 + 1}} 
\right)\mathrm{d}u \mathrm{d}v\\
&\stackrel{(a)}{=}
2 \pi 
\int_{0}^{\infty} 
-
\frac{\exp\left( - \frac{r^2}{\Theta_1 }\right)}{\pi \Theta_1 } 
\log_2
\left(
\frac{\exp\left( - \frac{r^2}{\Theta_1 }\right)}{\pi \Theta_1  \sqrt{4 \frac{k^2 \myP \Theta_2}{\Theta^2_1} r^2 + 1}} 
\right)r \mathrm{d} r\\
&= {\log_2 (\pi e \Theta_1)}+\frac{1}{2 \log(2)} \exp\left(\frac{\Theta_1 }{4 k^2 \myP \Theta_2}\right) \Ei \left(\frac{\Theta_1 }{4  k^2 \myP \Theta_2}\right),
\end{aligned}
\end{equation}
where $(a)$ is because of the polar transformation.

\subsubsection{Asymptotic $h(\tilde{A}_X + \tilde{A}_{\tilde{Z},N}\vert \tilde{A}_X)$}
For a given value of $\tilde{A}_X$, based on the definition of the PN coordinate system, the random variable $\tilde{A}_X + \tilde{A}_{\tilde{Z},N}$ is treated as the projection of the three-dimensional circular symmetric Gaussian noise, i.e., $(\frac{\tilde{Z}}{\sqrt{\myP}},k \frac{N}{\sqrt{\myP}})$ shifted by $\tilde{A}_X$, on the parabolic~$\mathcal{U}$ by the normal vectors of it.

As $\myP \rightarrow \infty$, $\tilde{A}_X +\tilde{A}_{\tilde{Z},N}$ converges in probability toward $\tilde{A}_X$, 
thus, for a given value of $\tilde{A}_X$, the effective range of the random variable $\tilde{A}_{X}+\tilde{A}_{\tilde{Z},N}$ on the paraboloid $\mathcal{U}$, is very small, which is close to the tangent plane of $\mathcal{U}$ at the point $\tilde{A}_{X}$. 
Therefore, 
the random variable $\tilde{A}_X +\tilde{A}_{\tilde{Z},N}$ converges in probability toward the random variable generated by the
projection of the three-dimensional circular symmetric Gaussian noise on the tangent plane of $\mathcal{U}$ at the point $\tilde{A}_{X}$ by the normal vector of the point $\tilde{A}_{X}$,
which is the well-known two-dimensional {\color{black}complex} Gaussian random variable with variance $\sigmaone/\myP$.
Therefore, given $\tilde{A}_X$, the entropy $h(\tilde{A}_X + \tilde{A}_{\tilde{Z},N})$ is approaching to $\log_2 (\pi e \sigmaone/\myP)$ which does not rely on $\tilde{A}_X$. Thus, as $\myP \rightarrow \infty$, the asymptotic conditional entropy is
\begin{equation} \label{part2}
h(\tilde{A}_X + \tilde{A}_{\tilde{Z},N} \vert \tilde{A}_X ) 
= \mathbb{E}_{\tilde{A}_X}\left[ h(\tilde{A}_X + \tilde{A}_{\tilde{Z},N} \vert \tilde{A}_X =\tilde{a}_X)  \right]
\approx \log_2 \left(\frac{\pi e \sigmaone}{\myP}\right).
\end{equation}

\subsubsection{Asymptotic $\mathcal{I}(\sqrt{\myP} X; \tilde{Y}_1, Y_2)$}	
Taking \eqref{highSNR_entropy_2} and \eqref{part2} into \eqref{MI_asymptotic_appen}, \eqref{theory1_1} is obtained.

\subsection{Proof of \eqref{theory1_2}}
Based on the power series expansion of the exponential integral function \cite{Handbook}
\begin{equation}
\Ei(x) = - \gamma - \ln x - \sum_{n=1}^{\infty} \frac{(-x)^n}{n\ n!},\ x >0,
\end{equation}
where $\gamma \approx 0.5772$ is Euler's constant,
as $\myP$ is sufficiently large, we have
\begin{equation} \label{asymp_expEi}
\lim\limits_{\myP \rightarrow \infty}  \exp\left(\frac{\Theta_1 \sigmatwo}{2 \sigmaone \myP \Theta_2}\right) \Ei\left(\frac{\Theta_1 \sigmatwo}{2 \sigmaone \myP \Theta_2}\right) = 
- \gamma + \ln \left(\frac{2 \sigmaone \myP \Theta_2}{\Theta_1 \sigmatwo}\right).
\end{equation}
Substituting \eqref{asymp_expEi} into \eqref{theory1_1}, \eqref{theory1_2} is obtained.

%
%
%

\end{appendices}

\ifCLASSOPTIONcaptionsoff
\fi
\bibliographystyle{IEEEtran}

\begin{thebibliography}{10}
	\providecommand{\url}[1]{#1}
	\csname url@samestyle\endcsname
	\providecommand{\newblock}{\relax}
	\providecommand{\bibinfo}[2]{#2}
	\providecommand{\BIBentrySTDinterwordspacing}{\spaceskip=0pt\relax}
	\providecommand{\BIBentryALTinterwordstretchfactor}{4}
	\providecommand{\BIBentryALTinterwordspacing}{\spaceskip=\fontdimen2\font plus
		\BIBentryALTinterwordstretchfactor\fontdimen3\font minus
		\fontdimen4\font\relax}
	\providecommand{\BIBforeignlanguage}[2]{{%
			\expandafter\ifx\csname l@#1\endcsname\relax
			\typeout{** WARNING: IEEEtran.bst: No hyphenation pattern has been}%
			\typeout{** loaded for the language `#1'. Using the pattern for}%
			\typeout{** the default language instead.}%
			\else
			\language=\csname l@#1\endcsname
			\fi
			#2}}
	\providecommand{\BIBdecl}{\relax}
	\BIBdecl
	
	\bibitem{BOOKTse}
	D.~Tse and P.~Viswanath, \emph{Fundamentals of wireless communication}.\hskip
	1em plus 0.5em minus 0.4em\relax Cambridge university press, 2005.
	
	\bibitem{MagazineRobert}
	A.~Alkhateeb, J.~Mo, N.~Gonzalez-Prelcic, and R.~W. Heath, ``{MIMO} precoding
	and combining solutions for millimeter-wave systems,'' \emph{{IEEE} Commun.
		Mag.}, vol.~52, no.~12, pp. 122--131, Dec. 2014.
	
	\bibitem{Jeffrey}
	\BIBentryALTinterwordspacing
	J.~G. {Andrews}, T.~{Bai}, M.~{Kulkarni}, A.~{Alkhateeb}, A.~{Gupta}, and R.~W.
	{Heath}, Jr, ``{Modeling and Analyzing Millimeter Wave Cellular Systems},''
	\emph{ArXiv e-prints}. [Online]. Available:
	\url{https://arxiv.org/pdf/1605.04283v1.pdf}
	\BIBentrySTDinterwordspacing
	
	\bibitem{Goldsmith16}
	M.~Chowdhury, A.~Manolakos, and A.~Goldsmith, ``Scaling laws for noncoherent
	energy-based communications in the {SIMO} {MAC},'' \emph{{IEEE} Trans. Inf.
		Theory}, vol.~62, no.~4, pp. 1980--1992, Apr. 2016.
	
	\bibitem{Petar16}
	L.~Jing, E.~D. Carvalho, P.~Popovski, and .~O. Martínez, ``Design and
	performance analysis of noncoherent detection systems with massive receiver
	arrays,'' \emph{{IEEE} Trans. Signal Process.}, vol.~64, no.~19, pp.
	5000--5010, Oct. 2016.
	
	\bibitem{Bi15}
	S.~Bi, C.~Ho, and R.~Zhang, ``Wireless powered communication: {Opportunities}
	and challenges,'' \emph{IEEE Commun. Mag.}, vol.~53, no.~4, pp. 117--125,
	Apr. 2015.
	
	\bibitem{Huang15}
	K.~Huang and X.~Zhou, ``Cutting the last wires for mobile communications by
	microwave power transfer,'' \emph{IEEE Commun. Mag.}, vol.~53, no.~6, pp.
	86--93, Jun. 2015.
	
	\bibitem{Krikidis_survey}
	I.~Krikidis, S.~Timotheou, S.~Nikolaou, G.~Zheng, D.~Ng, and R.~Schober,
	``Simultaneous wireless information and power transfer in modern
	communication systems,'' \emph{IEEE Commun. Mag.}, vol.~52, no.~11, pp.
	104--110, Nov. 2014.
	
	\bibitem{XiaoLu}
	X.~Lu, P.~Wang, D.~Niyato, D.~I. Kim, and Z.~Han, ``Wireless networks with {RF}
	energy harvesting: {A} contemporary survey,'' \emph{{IEEE} Commun. Surveys
		Tuts.}, vol.~17, no.~2, pp. 757--789, Second quarter 2015.
	
	\bibitem{InfraredMag}
	J.~M. Kahn and J.~R. Barry, ``Wireless infrared communications,'' \emph{Proc.
		{IEEE}}, vol.~85, no.~2, pp. 265--298, Feb. 1997.
	
	\bibitem{OpticModu}
	K.~Szczerba, P.~Westbergh, J.~Karout, J.~S. Gustavsson, {\AA}.~Haglund,
	M.~Karlsson, P.~A. Andrekson, E.~Agrell, and A.~Larsson, ``4-{PAM} for
	high-speed short-range optical communications,'' \emph{IEEE J. Opt. Commun.
		Netw.}, vol.~4, no.~11, pp. 885--894, Nov. 2012.
	
	\bibitem{OpticalTIT}
	A.~Lapidoth, S.~Moser, and M.~Wigger, ``On the capacity of free-space optical
	intensity channels,'' \emph{{IEEE} Trans. Inf. Theory}, vol.~55, no.~10, pp.
	4449--4461, Oct. 2009.
	
	\bibitem{split_circuit}
	Y.~Wu, Y.~Liu, Q.~Xue, S.~Li, and C.~Yu, ``Analytical design method of multiway
	dual-band planar power dividers with arbitrary power division,'' \emph{{IEEE}
		Trans. Microw. Theory Tech.}, vol.~58, no.~12, pp. 3832--3841, Dec. 2010.
	
	\bibitem{datasheet}
	\BIBentryALTinterwordspacing
	\emph{DC to 18 GHz Power Splitter}, Agilent Technologies. [Online]. Available:
	\url{http://cp.literature.agilent.com/litweb/pdf/5990-5351EN.pdf}
	\BIBentrySTDinterwordspacing
	
	\bibitem{Xunzhou13}
	X.~Zhou, R.~Zhang, and C.~K. Ho, ``Wireless information and power transfer:
	{Architecture} design and rate-energy tradeoff,'' \emph{{IEEE} Trans.
		Commun.}, vol.~61, no.~11, pp. 4754--4767, Nov. 2013.
	
	\bibitem{WanchunICC16}
	W.~Liu, X.~Zhou, S.~Durrani, and P.~Popovski, ``{SWIPT} with practical
	modulation and {RF} energy harvesting sensitivity,'' in \emph{Proc. IEEE
		ICC}, May 2016, pp. 1--7.
	
	\bibitem{BookInfo}
	T.~Cover and J.~Thomas, \emph{Elements of Information Theory}.\hskip 1em plus
	0.5em minus 0.4em\relax Wiley, 2006.
	
	\bibitem{EMG}
	E.~Grushka, ``Characterization of exponentially modified {Gaussian} peaks in
	chromatography,'' \emph{Analytical Chemistry}, vol.~44, no.~11, pp.
	1733--1738, 1972.
	
	\bibitem{Invariance}
	A.~Kraskov, H.~St{\"o}gbauer, and P.~Grassberger, ``Estimating mutual
	information,'' \emph{Physical review E}, vol.~69, no.~6, p. 066138, 2004.
	
	\bibitem{Handbook}
	M.~Abramowitz and I.~A. Stegun, \emph{Handbook of mathematical functions: with
		formulas, graphs, and mathematical tables}.\hskip 1em plus 0.5em minus
	0.4em\relax Courier Corporation, 1964, vol.~55.
	
\end{thebibliography}

\end{document}